\documentclass[12pt]{article}
\usepackage{UF_FRED_paper_style}
\usepackage{times}
\usepackage{amsmath, amssymb}
\usepackage{geometry}
\usepackage{cancel}
\usepackage{marginnote}
\usepackage{amsthm}
\newtheorem{lemma}{Lemma}
\newtheorem{definition}{Definition}
\usepackage{cleveref}
\usepackage{natbib}
\usepackage{appendix}
\newcommand{\Beta}{\mathrm{Beta}}
\usepackage{booktabs}

\usepackage{enumitem}
\hypersetup{
    colorlinks = true,
    citecolor = blue
}
\usepackage{lipsum}  
\usepackage{setspace}
\newtheorem{theorem}{Theorem}

\usepackage{bm}
\usepackage{booktabs}
\usepackage{subcaption}
\usepackage{graphicx}
\title{\textbf{Optimal sequential two-stage Bayes Factor Design for two-arm clinical Phase II Trials with binary Endpoints}}

\author{Riko Kelter\thanks{Correspondence concerning this article should be addressed to \url{rkelter@uni-koeln.de}.
    Draft version 1.0, 01/06/26. This paper has not been peer reviewed. Please do not copy or cite without author's permission. The R package \texttt{bfbin2arm} is available on CRAN, see \url{https://cran.r-project.org/web/packages/bfbin2arm/index.html}. Data and R code to reproduce all results are openly available at the associated two-stage two-arm vignette at CRAN. The author declares no conflict of interest.}\\
	Institute of Medical Statistics and Computational Biology\\
	Faculty of Medicine\\
    University of Cologne\\
    Cologne, Germany
    }

\date{\today}

\begin{document}

{\setstretch{.8}
\maketitle
\begin{abstract}
Two-arm phase II clinical trials often benefit from an interim analysis that allows early stopping for futility, but Bayesian calibration of such designs is usually based on computationally intensive Monte Carlo simulation. In this work, a simulation-free methodology is developed to obtain Bayesian optimal two-stage designs in two-arm phase II trials with binary endpoints using Bayes factors as the primary measure of evidence. Building on recent matrix-search methods for fixed-sample two-arm Bayes factor designs and earlier correction formulas for one-arm two-stage designs, the proposed approach derives exact expressions for the operating characteristics of a two-stage two-arm design with a single futility interim. Bayesian power and type-I error are obtained by correcting the corresponding fixed-sample quantities for trajectories that would have been removed by early stopping, yielding a fully numerical calibration procedure that avoids Monte Carlo error entirely. The resulting method searches over admissible interim and final sample sizes to identify the optimal design that satisfies target constraints on Bayesian power, type-I error, and the probability of compelling evidence in favour of the null hypothesis, while minimizing the expected sample size under the null hypothesis. The methodology is illustrated in realistic phase II settings, including a detailed re-analysis of the riociguat trial in systemic sclerosis. Overall, the approach extends simulation-free Bayes factor design methodology to the practically important setting of two-arm two-stage phase II trials and provides a transparent basis for Bayesian design calibration and sensitivity analysis.


\noindent
\textit{\textbf{Keywords: }%
phase II trial, sequential design, two-stage design, optimal trial design, Bayesian statistics, Bayes factors, two-arm clinical trial, binary endpoint} \\ 
\noindent

\end{abstract}
}

\section{Introduction}

Two-arm phase~II clinical trials with binary endpoints are a central tool for assessing the preliminary efficacy of novel treatments before moving to larger, confirmatory phase~III studies \citep{Chow2008,Spiegelhalter2004,Grieve2022}. Classical designs, such as Simon's two-stage procedure, provide explicit frequentist control of type-I and type-II error rates and offer the option of early stopping for futility, thereby improving both ethical and efficiency properties compared to fixed-sample designs \citep{Simon1989,WassmerBrannath2016}. However, these methods typically do not incorporate prior information in a principled way and are not framed in terms of coherent Bayesian measures of evidence.

Bayesian designs address these shortcomings by allowing the formal inclusion of historical data, expert opinion, or mechanistic knowledge via prior distributions, and by updating beliefs as data accrue \citep{Berry2006,Spiegelhalter2004,Thall1994,Neuenschwander2009}. Decision rules are often based on posterior probabilities or Bayes factors, with the latter quantifying the relative support of the data for competing hypotheses in a way that is invariant to the prior odds on the hypotheses themselves \citep{Jeffreys1939,KassRaftery1995,Rouder2009,VanDeSchoot2021}. As regulatory agencies increasingly encourage Bayesian analyses that demonstrate acceptable frequentist operating characteristics, so-called calibrated Bayes or Bayes–frequentist compromise approaches have gained prominence \citep{Dawid1982,Little2006,grieveIdleThoughtsWellcalibrated2016,FDA_ComplexInnovativeDesignsDecember2020,ionanBayesianMethodsHuman2023}. For a recent guidance for industry issued by the Food and Drug Administration (FDA) on the use of Bayesian methodology in clinical trials, see \cite{FDA_UseOfBayesianMethodologyJanuary2026}.

In practice, however, Bayesian sample size planning and design calibration often rely on intensive Monte Carlo simulation studies. For a given Bayesian test statistic---such as a Bayes factor or a posterior probability---and a chosen decision threshold, power and type-I-error rates are typically evaluated by simulating many trial replicates under $H_0$ and $H_1$ \citep{Berry2011,Schonbrodt2017,Stefan2022,Grieve2022}. This simulation-based paradigm raises several challenges:
\begin{enumerate}[label=(\roman*)]
    \item{calibration becomes computationally expensive}
    \item{reproducibility depends on reporting Monte Carlo standard errors and implementation details}
    \item{every change in priors, thresholds, or design parameters may require re-running large simulations}
\end{enumerate}
For more details, see also \cite{Morris2019}, \cite{Boulesteix2020}, \cite{Kelter2023}, and \cite{siepeSimulationStudiesMethodological2024}.

There are several approaches available in the literature which either aim at reducing the computational burden associated with calibrating a Bayesian design or at least providing a calibrated Bayesian design, even if the computational effort often is substantial. A comprehensive review of Bayesian sequential clinical trial designs based on posterior and predictive probabilities, as well as decision‑theoretic criteria, is given by \cite{Zhou2023}, who also discuss frequentist, calibrated Bayesian, and subjective Bayesian perspectives on interim monitoring and the likelihood principle. 

In a related attempt to reduce the computational burden of calibrating Bayesian designs, \cite{hagarDesignBayesianClinical2026} propose an efficient methodology for Bayesian clinical trials with clustered data that models posterior probabilities as functions of the number of clusters to assess operating characteristics across sample sizes from only a few simulation points.

Similarly, \cite{zhuBayesianSequentialDesign2019} propose a Bayesian sequential design for time‑to‑event outcomes that uses alpha‑spending functions to control the overall type‑I error rate and employs Bayes factors for interim decision‑making, illustrating that Bayes factor–based sequential monitoring can match or improve the efficiency of classical group sequential designs.

Another example is \cite{gaoBayesianSequentialDecisionmaking2025}, who develop a Bayesian sequential decision\-‑making framework for rare disease trials with binary endpoints, combining sequential Bayes factor updates with adaptive stopping rules for superiority and futility to reduce expected sample size while maintaining interpretable evidence thresholds.

\cite{shenBayesianGroupSequential2022} consider Bayesian group sequential designs for cluster‑randomized trials, proposing flexible schemes that allow early stopping for efficacy at pre‑planned interim analyses and exploring their operating characteristics via simulation for different recruitment patterns and outcome types.

However, all of these approaches rely on simulating the trial operating characteristics in one form or another, leading to the problems (i) to (iii) described above.

Recent work has shown that these obstacles can be overcome in important special cases. In the one-arm binomial setting, numerical root-finding and prior-predictive calculations allow for essentially instantaneous Bayesian power and sample size computations for Bayes factors, entirely avoiding Monte Carlo simulation \citep{KelterPawel2025}. This approach has been extended to a Bayesian optimal two-stage design for single-arm phase~II trials with binary endpoints, where a single interim analysis is accommodated via a trinomial-tree representation of the Bayes factor trajectories, and the resulting power and type-I error rates are analytically corrected for the possibility of early stopping \citep{KelterPawelTwoStage2025}. More recently, analogous matrix-search methods have been developed for two-arm binomial phase~II designs, yielding simulation-free Bayesian power and sample size calculations for a broad class of Bayes factors in the two-arm setting \citep{kelterTwoArmTwoStage2026}. We provide details and summarize these approaches in \Cref{sec:background}.

Despite these advances, there is currently no simulation-free methodology for Bayesian two-arm phase~II designs with both (i) Bayes factor-based decision rules and (ii) a formal two-stage structure allowing for a single interim analysis.\footnote{A notable exception is the recent work of \cite{pawelBayesFactorGroup2026}, who extend classical group sequential theory to Bayes factor designs by mapping Bayes factor stopping regions to z‑statistic boundaries, allowing fast, simulation‑free computation of stopping probabilities via multivariate normal integration. The approach proposed in this paper shares the same goal, in the sense that it is simulation-free but focusses on binomial endpoints. In contrast to the approach of \cite{pawelBayesFactorGroup2026}, we do not map Bayes factor stopping regions to z-statistic boundaries. Also, no use of classical group sequential theory and its asymptotic arguments is made in this paper, which might in some cases become problematic in the context of a clinical phase II trial due to its limited sample size.} The aim of this work is to fill this gap by combining the trinomial-tree correction ideas from the one-arm two-stage design with the matrix-search framework for two-arm Bayes factor calibration, thereby providing a fully numerical, simulation-free approach to Bayesian two-stage two-arm phase~II trial design with binary endpoints.

\section{Outline}
\label{sec:outline}

The remainder of this manuscript is organized as follows. 
Section~\ref{sec:background} reviews Bayesian power and sample size calculations for Bayes factors in binomial models, summarizing existing simulation-free results for one-arm fixed-sample designs, one-arm two-stage designs, and two-arm fixed-sample designs. 
Section~\ref{sec:method} introduces the proposed two-stage two-arm Bayes factor design, detailing the construction of interim and final decision regions, the corresponding prior-predictive probabilities, and the correction of power and type-I-error for early stopping. 
We derive several main results, based on which Section~\ref{subsec:calibration} then presents the calibration algorithm for choosing interim and final sample sizes for an optimal Bayesian design under prespecified operating characteristics. Section~\ref{subsec:calibration} also discusses optimization criteria such as minimizing the expected sample size under the null hypothesis to classify a design as optimal from a Bayesian point of view.
Section~\ref{sec:examples} illustrates the method in realistic phase~II scenarios, and Section~\ref{sec:discussion} concludes with a discussion of practical implications, limitations, and directions for future research.

\section{Background}
\label{sec:background}

\subsection{Bayes factors and calibrated Bayesian design}

Bayes factors quantify the relative evidence provided by the data for two competing hypotheses $H_0$ and $H_1$ via the ratio of their marginal likelihoods \citep{Jeffreys1939,KassRaftery1995}:
\[
  BF_{01}(y) 
  = \frac{f(y \mid H_0)}{f(y \mid H_1)}.
\]
Interpreted as a predictive updating factor from prior to posterior odds,
\begin{align}\label{eq:bayesFactorIntro}
   \underbrace{\frac{P(H_0 \mid y)}{P(H_1 \mid y)}}_{\text{Posterior odds}} = 
   \underbrace{\frac{f(y \mid H_0)}{f(y \mid H_1)}}_{\text{Bayes factor $\mathrm{BF}_{01}(y)$}} \cdot 
   \underbrace{\frac{P(H_0)}{P(H_1)}}_{\text{Prior odds}},
\end{align}
Bayes factors separate the influence of the prior odds $P(H_0)/P(H_1)$ on the hypotheses from the influence of the parameter priors within each hypothesis. In particular, for a fixed pair of design and analysis priors on the model parameters, the Bayes factor reflects only how the data update relative support for $H_0$ vs.\ $H_1$ \citep{VanDeSchoot2021,Bartos2022,Kelter2020BayesianPosteriorIndices,Good1983a,Kelter2022EvidenceValue}. This separation has been argued to make Bayes factors a more transparent index of evidence than posterior probabilities, especially when the prior odds on the hypotheses are themselves controversial or informed by historical information \citep{Grieve2022,Kelter2020,Linde2020,Makowski2019,Kelter2021BMCHodgesLehmann}.

From a design perspective, Bayes factors can be used as test statistics for which frequentist-style operating characteristics such as power and type-I-error are defined in terms of exceedance probabilities of Bayes factor thresholds.\footnote{This approach was championed by \cite{Good1983a}, who proposed it as a Bayes-frequentist compromise. Harold Jeffreys already suggested using what we now call Bayes factors as test statistics and studying their long‑run behavior in his 1939 monograph and subsequent editions of \textit{Theory of Probability} \citep{Jeffreys1939}. His methodology explicitly ties Bayes factors to error‑rate style criteria and proposes fixed thresholds (Jeffreys' scale) in analogy to significance testing. Jack Good then developed this line further in the 1960s–1980s, discussing ``weight of evidence'' (essentially the logarithm of the Bayes factor) and emphasizing that Bayes factor–based tests can be judged by their long‑run frequencies of misleading evidence and related error probabilities. He is widely cited as an early advocate of calibrating Bayesian procedures (including Bayes factors) via long‑run error concepts. For an overview see \cite{sekulovskiGoodCheckBayes2024}.} For example, using the $BF_{01}$ orientation and an evidence threshold $k<1$, one may define
\begin{align}
  &\text{Bayesian type-I-error: }
  P\bigl(BF_{01}(Y) < k \mid H_0\bigr),\\
  &\text{Bayesian power: }
  P\bigl(BF_{01}(Y) < k \mid H_1\bigr),
\end{align}
and require that these quantities satisfy inequalities analogous to classical design constraints for prespecified $\alpha,\beta\in(0,1)$, such as
\begin{align}\label{eq:t1e_target_constraint}
    P(BF_{01}(Y)< k\mid H_0)\le \alpha
\end{align}
and
\begin{align}\label{eq:power_target_constraint}
    P(BF_{01}(Y)< k\mid H_1)\ge 1-\beta
\end{align}
\citep{KelterPawel2025,grieveIdleThoughtsWellcalibrated2016,Grieve2022,pourmohamadSequentialBayesFactors2023}. This Bayes–frequen\-tist compromise is attractive for trial planners and regulators: it preserves a coherent Bayesian evidence measure while guaranteeing interpretable long-run error control \citep{Dawid1982,Little2006,Grieve2022,ionanBayesianMethodsHuman2023}. This goes beyond the current recommendations and advice of regulatory agencies like the FDA or European Medicine's Agency (EMA), compare \cite{FDA_ComplexInnovativeDesignsDecember2020,FDA_UseOfBayesianMethodologyJanuary2026,europeanmedicinesagencyICHE20Adaptive2025}.

\subsection{Simulation-free Bayes factor calibration in the binomial setting}

In most realistic models, the distributions of Bayes factors under $H_0$ and $H_1$ do not admit closed-form expressions, and power and type-I-error must be evaluated via Monte Carlo simulation \citep{Berry2011,Schonbrodt2017,Stefan2022}. This is particularly true for adaptive or sequential designs, where the stopping rule and updating scheme introduce complex dependencies across interim looks \citep{Chevret2012,Zhou2023,kelterBayesianGroupSequentialPredictive2024}. Simulation-based calibration, however, suffers from several drawbacks: it is computationally expensive, sensitive to the choice of simulation size and random seeds, and requires careful reporting of Monte Carlo standard errors and convergence diagnostics to ensure reproducibility \citep{Morris2019,Boulesteix2020,Kelter2023}.

In the one-arm binomial setting, \citet{KelterPawel2025} showed that these limitations can be circumvented. Focusing on tests of $H_0: p=p_0$ versus $H_1: p\neq p_0$ or directional alternatives, they derive Bayes factors under conjugate beta priors and obtain simple expressions for the prior-predictive distribution of the binomial count. Power and type-I-error are then written as sums of prior-predictive probabilities over sets of critical values determined by numerical root-finding on the Bayes factor. In other words, the Monte Carlo step is replaced by a direct evaluation of
\[
  P\bigl(BF_{01}(Y)< k \mid H_i\bigr)
  = \sum_{y \in \mathcal{Y}_k} f(y \mid H_i),
  \qquad i\in\{0,1\},
\]
where $\mathcal{Y}_k$ is the set of counts at which the Bayes factor crosses the evidence threshold $k$. This approach leads to essentially instantaneous computation of Bayesian power and sample size in the binomial setting, without any simulation and without relying on asymptotic approximation. It is implemented in the R package \texttt{bfbin2arm} as well as in the R package \texttt{bfpwr} \citep{PawelHeld2025}.

\begin{figure}[!htb]
    \centering
    \includegraphics[width=1.0\linewidth]{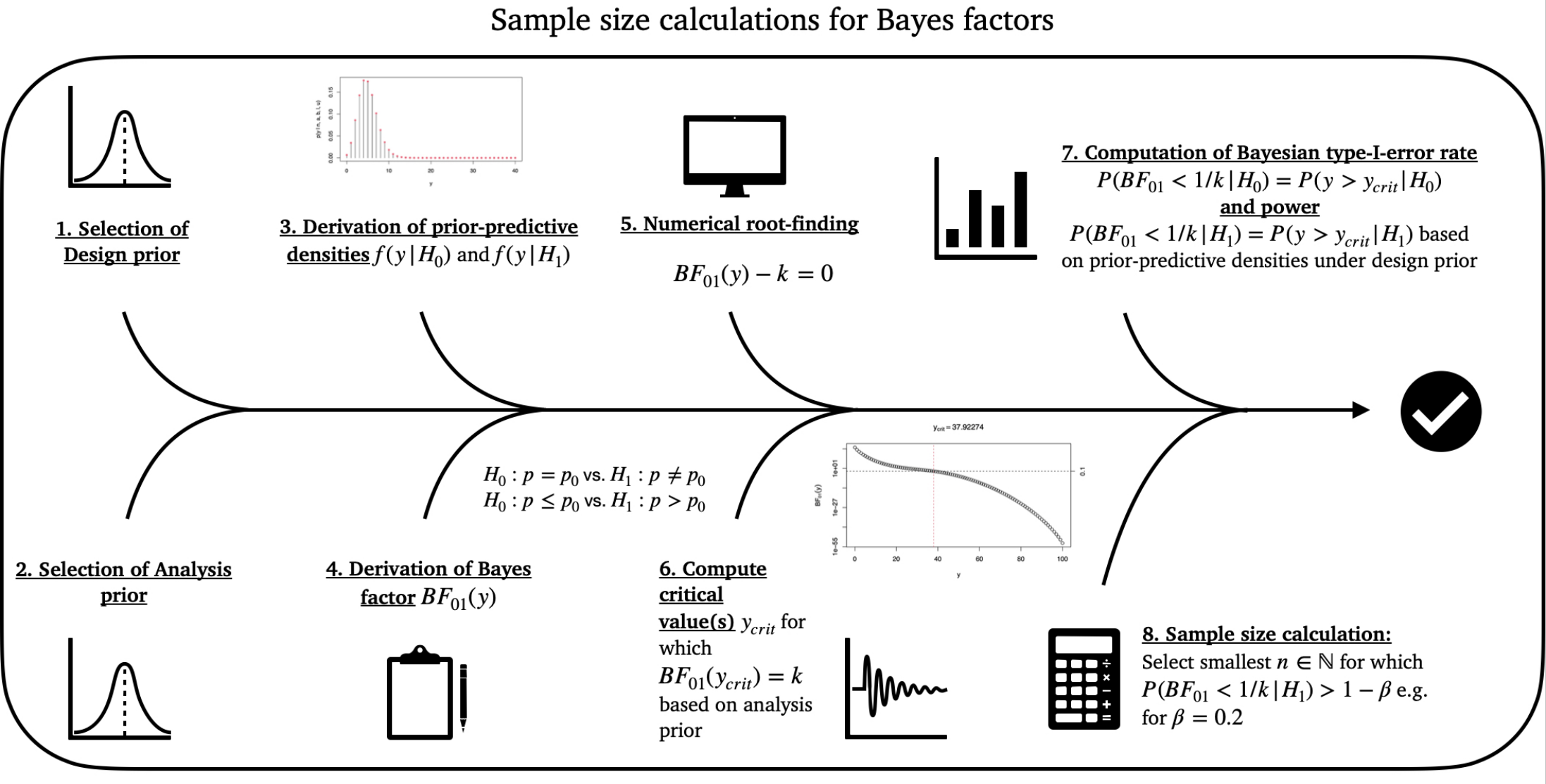}
    \caption{Overview of Bayesian power and sample size calculations for the case of a single-arm phase II trial with a binary endpoint, using Bayes factors. Details are provided in \cite{KelterPawel2025}.}
    \label{fig:flowchart_rootfinding}
\end{figure}
\Cref{fig:flowchart_rootfinding} visualizes the process of Bayesian power and sample size calculations for the single-arm phase II trial case with a binary endpoint, compare \cite{KelterPawel2025}.

\subsection{One-arm two-stage Bayes factor designs via trinomial-tree branching}

Building on this root-finding framework, \citet{KelterPawelTwoStage2025} proposed a Bayesian optimal two-stage design for single-arm phase~II trials with binary endpoints based on Bayes factors. The design introduces a single interim analysis after $n_1$ patients, with the option to stop early for futility if the Bayes factor indicates strong evidence in favour of $H_0$. The novelty lies in showing how to correct the Bayesian power and type-I-error rate for the presence of this interim look, again without resorting to Monte Carlo simulation.

\begin{figure}[h!]
    \centering
    \includegraphics[width=1\linewidth]{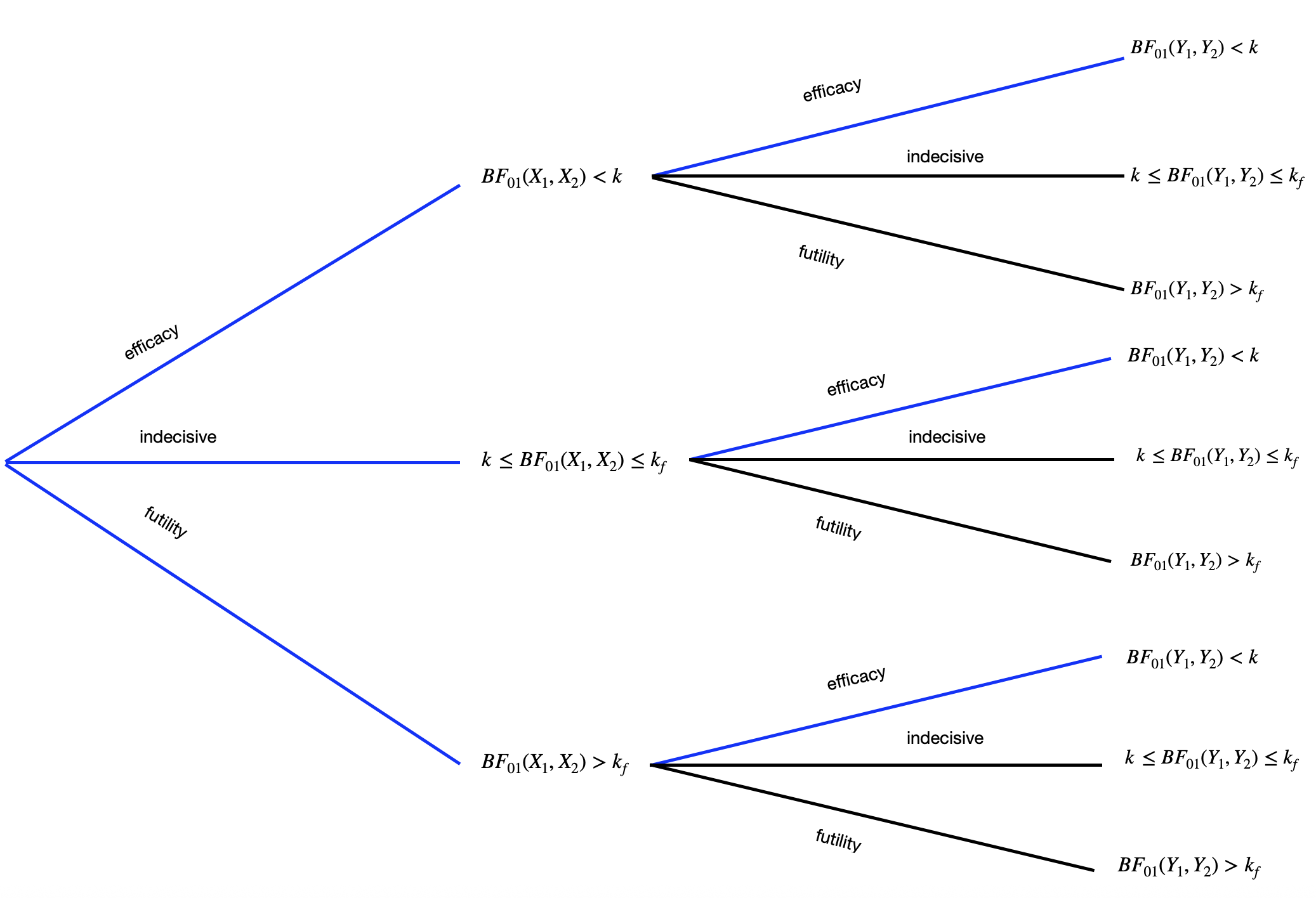}
    \caption{Trinomial tree underlying the two‑stage Bayesian Bayes‑factor design. At the interim analysis after $n_1$ patients, the trial may stop for futility (when $BF_{01}(X_1,X_2)>k_f$, indicating evidence in favour of the null hypothesis) or continue to the final analysis after $n_2$ patients (when $k\leq BF_{01}(X_1,X_2)\leq k_f$, indicating an indecisive result, or when $BF_{01}(X_1,X_2)<k$, indicating efficacy). At the final analysis, the same three decisions are possible, based on the Bayes factor $BF_{01}(Y_1,Y_2)$ and the thresholds $k$ for efficacy and $k_f$ for futility. Blue trajectories show all outcomes which contribute to Bayesian power in the sense of concluding that $H_1$ holds at the end of the trial.}
    \label{fig:twoStageDesign}
\end{figure}

The key device is a trinomial-tree representation of the Bayes factor trajectories: at each analysis (interim and final), the Bayes factor can indicate efficacy (evidence for $H_1$), futility (evidence for $H_0$), or be inconclusive, compare \Cref{fig:twoStageDesign}. In \Cref{fig:twoStageDesign}, $(X_1,X_2)$ denotes the data available in the treatment and control group at interim analysis and $BF_{01}(X_1,X_2)$ the Bayes factor based on the available interim data. The tupel $(Y_1,Y_2)$ denotes the full trial data available at the end of the trial and $BF_{01}(Y_1,Y_2)$ the Bayes factor based on this full trial data. We introduce the setup and notation in detail in \Cref{subsec:notationAndSetup}.\footnote{In \cite{KelterPawelTwoStage2025}, the one-arm setting with only a treatment group is considered, so there the Bayes factors are based on treatment group data $Y$ only. Here, we consider the two-arm phase II setting with a treatment and control group, and modified notation and \Cref{fig:twoStageDesign} accordingly.}

When the design is naively calibrated using the fixed-sample expressions at $n_2$ alone, trajectories that would have stopped for futility at the interim but later ``swing back'' to evidence for $H_1$ at the final analysis are incorrectly counted as contributing to power (and analogously for type-I-error under $H_0$). The authors identify these trajectories as a ``futility-erased partial power'' (and ``futility-erased partial type-I-error'') and provide closed-form summation formulas, using prior-predictive distributions, to subtract these contributions from the fixed-sample power and type-I-error \citep{KelterPawelTwoStage2025}. This yields corrected operating characteristics for the two-stage design that account exactly for the possibility of early stopping, while preserving the simulation-free nature of the calibration.

In terms of \Cref{fig:twoStageDesign}, this corresponds to the lowest blue trajectory: Without an interim analysis, the Bayes factor could indicate futility when calculated based on the interim sample size $n_1$, and then swing back to reach efficacy in the final analysis. These trajectories contribute to Bayesian power when no interim analysis is carried out. Once an interim analysis is introduced, however, the trial can be stopped when the futility threshold is reached at the interim sample size, reducing the power by ``cutting off'' the possibility that the Bayes factor swings around and reaches efficacy for the final sample size at the end of the trial in these trajectories.

In addition to the correction the authors propose to solve this problem, \citet{KelterPawelTwoStage2025} develop a calibration algorithm that searches $(n_1,n_2)$ to find Bayesian optimal two-stage designs that
\begin{itemize}
    \item[(i)]{satisfy prespecified constraints on Bayesian power and type-I-error, compare \Cref{eq:power_target_constraint} and \Cref{eq:t1e_target_constraint}, and}
    \item[(ii)]{minimize the expected sample size $E[N|H_0]$ under $H_0$. The resulting design thus is computed as the result of the following optimization problem:
    \begin{equation}\label{eq:minExpN}
        \begin{aligned}
            \min_{n_1,n_2} \quad & E[N|H_0]\\
            \textrm{subject to} \quad & P(\mathrm{BF}_{01}^{n_2}(y) < k \mid H_0) \leq \alpha \hspace{1cm}\\
            & \text{and} \hspace{1cm} P(\mathrm{BF}_{01}^{n_2}(y) < k \mid H_1) \geq 1-\beta\\
            & \text{and} \hspace{1cm} n_{min}\leq n_1 < n_2\leq n_{max}    \\
        \end{aligned}
    \end{equation}}
\end{itemize}
where $\mathrm{BF}_{01}^{n_2}(y)$ is the Bayes factor based on the final sample size $n_2$ at the end of the trial. The resulting designs recover Simon-type optimal designs as special cases, improve non-sequential Bayes factor designs, and can be calibrated rapidly using only standard numerical methods.\footnote{The calibration algorithm for the two-stage single-arm design is currently implemented in the \texttt{bfbin2arm} R package \citep{kelterTwoArmTwoStage2026}, available on CRAN under \url{https://cran.r-project.org/web/packages/bfbin2arm/index.html}.} \cite{KelterPawel2025} call such a design \textit{optimal} in the Bayesian sense.

\subsection{Two-arm fixed-sample designs via matrix search}
\label{subsec:matrix-search}
\citet{kelterTwoArmTwoStage2026} extended the root-finding approach to the two-arm binomial setting, accommodating a variety of hypotheses relevant for two-arm phase~II trials, including equality of response probabilities $H_0: p_1 = p_2$ vs.\ $H_1: p_1 \neq p_2$), superiority ($H_0: \eta \le 0$ vs.\ $H_1: \eta > 0$ -- where $\eta:=p_2-p_1$) denotes the difference in success probabilities between the treatment and control arm -- and ordered alternatives. The corresponding Bayes factors is derived there under flexible beta design and analysis priors as well as the joint prior-predictive distribution of the binomial counts in the two arms, $(Y_1,Y_2)$, in closed form.

The key insight is that the discrete nature of the binomial counts reduces the problem to a finite \emph{matrix search} over all integer pairs $(y_1,y_2) \in \{0,\dots,n_1\} \times \{0,\dots,n_2\}$. For fixed sample sizes $n_1,n_2$ in the control and treatment arms and a chosen evidence threshold $k <1$, one first computes the Bayes factor $BF_{01}(y_1,y_2)$ at every lattice point. The efficacy region---the set of count pairs that provide evidence against $H_0$---is then
\[
  \mathcal{E}_2 
  = \bigl\{(y_1,y_2): BF_{01}(y_1,y_2)< k\bigr\}.
\]

Figure~\ref{fig:matrix-search} illustrates this procedure for $n_1 = n_2 = 5$ and flat analysis priors ($\alpha_i^a = \beta_i^a = 1$). The left panel shows the matrix of Bayes factors $BF_{01}(y_1,y_2)$, where rows index the number of control arm successes $y_1$ and columns index treatment arm successes $y_2$. As expected under the two-sided test $H_0: p_1 = p_2$ vs.\ $H_1: p_1 \neq p_2$, the matrix is symmetric around the main diagonal ($y_1 = y_2$), with largest values (strongest evidence for $H_0$) along this diagonal and progressively smaller values toward the top-right and bottom-left margins (evidence against $H_0$).

\begin{figure}[htbp]
\centering
\begin{tabular}{ccc}
  \begin{minipage}{0.33\textwidth}
    \centering
    \includegraphics[width=\textwidth]{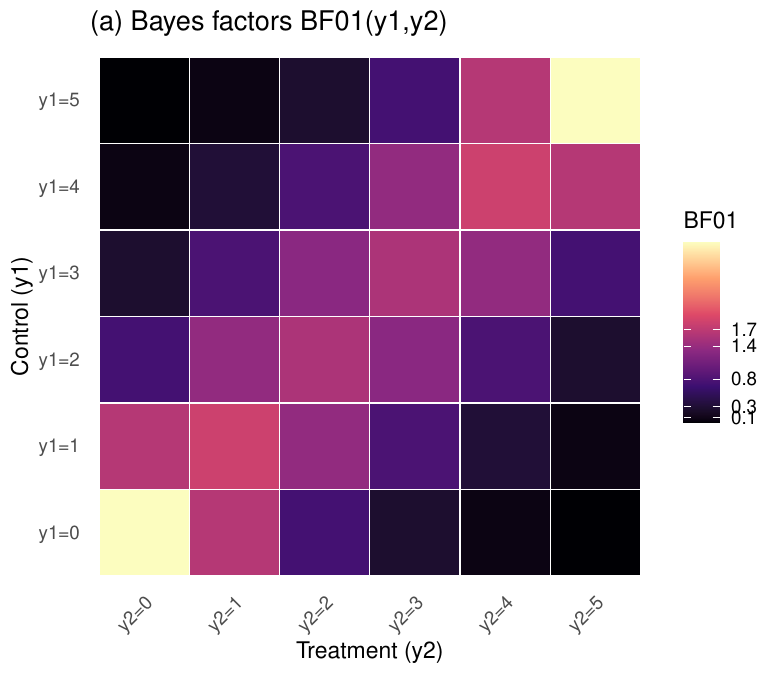}
    \small (a) Bayes factors $BF_{01}(y_1,y_2)$
  \end{minipage}
  &
  \begin{minipage}{0.33\textwidth}
    \centering
    \includegraphics[width=\textwidth]{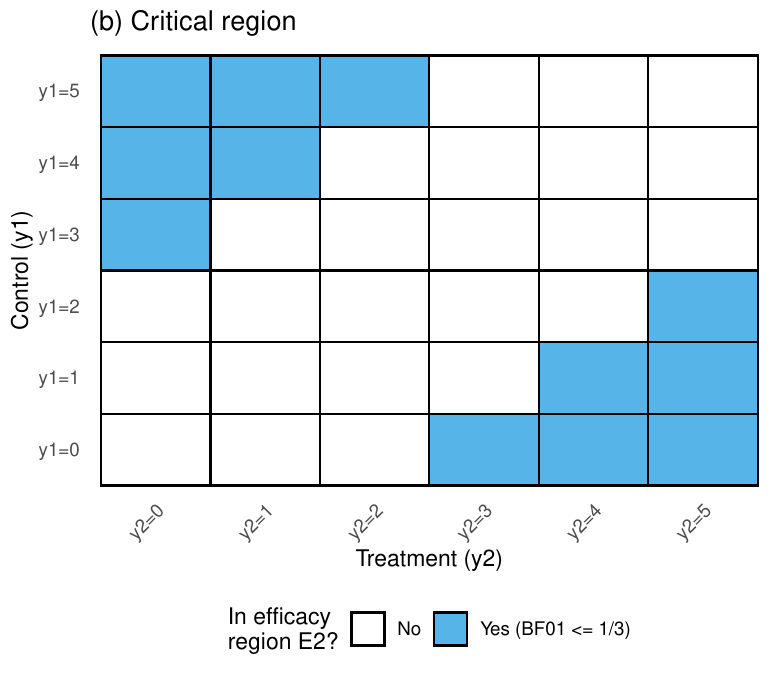}
    \small (b) Critical region: $BF_{01} < 1/3$
  \end{minipage}
  &
  \begin{minipage}{0.33\textwidth}
    \centering
    \includegraphics[width=\textwidth]{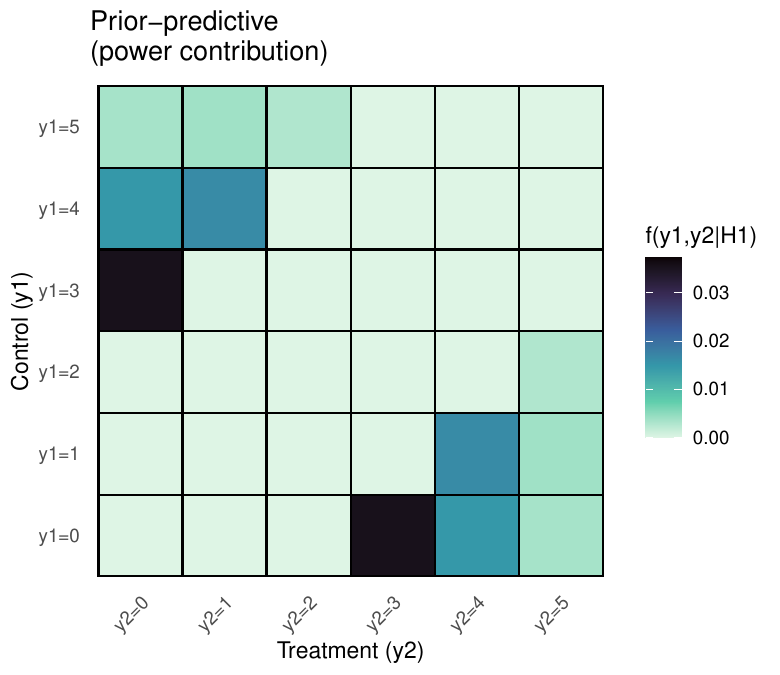}
    \small (c) Prior-predictive under $H_1$
  \end{minipage}
\end{tabular}
\caption{Matrix-search procedure for two-arm Bayes factor power calculation ($n_1=n_2=5$, flat analysis priors, $k=1/3$). (a) Full matrix of Bayes factors $BF_{01}(y_1,y_2)$. (b) Critical efficacy region $\mathcal{E}_2 = \{(y_1,y_2): BF_{01}(y_1,y_2) < 1/3\}$ (zeros elsewhere). (c) Prior-predictive probabilities $f(y_1,y_2\mid H_1)$ over $\mathcal{E}_2$; Bayesian power is their sum.}
\label{fig:matrix-search}
\end{figure}

The middle panel (b) identifies the critical efficacy region by setting all entries where $BF_{01}(y_1,y_2) > 1/3$ to zero, leaving only those lattice points that would lead to rejection of $H_0$. In this example,
\[
  \mathcal{E}_2 = \{(0,3),(0,4),(0,5),(1,4),(1,5),(2,5),(3,0),(4,0),(4,1),(5,0),(5,1),(5,2)\}.
\]
Finally, the right panel (c) shows the prior-predictive probabilities $f(y_1,y_2\mid H_1)$ under a chosen design prior, restricted to $\mathcal{E}_2$. The Bayesian power is simply the sum of these probabilities:
\[
  P\bigl(BF_{01}(Y_1,Y_2) < k \mid H_1\bigr)
  = \sum_{(y_1,y_2) \in \mathcal{E}_2} f(y_1,y_2 \mid H_1) \approx 0.33.
\]
The type-I-error rate under $H_0$ is computed analogously by summing $f(y_1,y_2\mid H_0)$ over the same region $\mathcal{E}_2$.

This matrix-search approach completely replaces both numerical root-finding (one-arm case) and Monte Carlo simulation with direct enumeration of the finite sample space. Sample size determination proceeds by repeating the procedure for increasing $(n_1,n_2)$ until the desired power and type-I-error bounds are achieved. The methodology is implemented in the \texttt{bfbin2arm} R package, enabling rapid exploration of two-arm Bayes factor designs with a treatment and control group across a wide range of hypotheses, priors, and thresholds \citep{kelterTwoArmTwoStage2026}.

\subsection{Motivation for a two-arm two-stage Bayes factor design}

The three strands of work described above establish a coherent simulation-free framework for Bayes factor-based designs in (i) one-arm fixed-sample binomial trials \citep{KelterPawel2025}, (ii) one-arm two-stage trials via trinomial-tree branching and futility-erased corrections \citep{KelterPawelTwoStage2025}, and (iii) two-arm fixed-sample binomial trials via matrix search \citep{kelterTwoArmTwoStage2026}. Together, they demonstrate that Bayesian designs with Bayes factor decision rules can be calibrated numerically, without Monte Carlo simulation, while maintaining interpretable Bayesian and frequentist properties.

In many phase~II settings, however, the combination of a control and treatment arm and interim monitoring is highly desirable. Two-arm designs offer a more realistic assessment of treatment effect by directly comparing a novel treatment to standard of care or placebo, while two-stage designs with an interim futility analysis reduce expected sample size under $H_0$ and limit exposure to ineffective therapies \citep{Simon1989,WassmerBrannath2016,Berry2011}. Existing Bayesian two-arm designs with interim analyses usually rely on simulation-based calibration and often do not use Bayes factors as the primary decision criterion \citep{Fayers2005,Ferguson2021,Stallard2020,Ferreira2021}.\footnote{Nothing is requiring to use Bayes factor as the test statistic for assessing the hypotheses under consideration in the previous work outlined in \Cref{sec:background}. Here, we focus on using Bayes factors due to their advantages as a measure of statistical evidence. For a detailed treatment of the advantages and limitations of various measures of statistical evidence see \cite{Sprenger2019}. The important implication for the methodology developed in the current manuscript is that one could use all of the power and sample size calculation methodology developed so far and adjust the measure of statistical evidence to e.g. posterior probabilities or posterior odds. This would allow to apply the current methodology developed in this paper also to these measures of statistical evidence, yielding optimal designs which are based on e.g. posterior probabilities or posterior odds of the competing hypotheses.}

The natural next step, therefore, is to combine the trinomial-tree correction ideas of the one-arm two-stage Bayes factor design with the matrix-search framework of the two-arm fixed-sample design. Conceptually, this entails moving from (i) a one-dimensional count of successes in a single arm to (ii) a two-dimensional grid of successes in two arms, and from (iii) a two-dimensional sum over interim and final counts in the one-arm case to (iv) a four-dimensional sum over interim and incremental counts in the two-arm case. The same logic applies: identify the set of trajectories that would be cut off by an interim futility stopping rule but that would otherwise contribute to fixed-sample power or type-I-error, and subtract their prior-predictive probabilities from the naive fixed-sample operating characteristics.

By doing so, one can derive a simulation-free, Bayes factor-based two-stage design for two-arm phase~II trials with binary endpoints that:
\begin{itemize}
  \item Controls Bayesian analogues of type-I-error and power at prespecified levels, in a way that is interpretable from both Bayesian and frequentist perspectives, compare \Cref{eq:t1e_target_constraint} and \Cref{eq:power_target_constraint}.
  \item Admits explicit numerical corrections for the interim futility analysis via sums of prior-predictive probabilities, avoiding Monte Carlo simulations.
  \item Allows optimization criteria such as minimal expected sample size under $H_0$, generalizing the notion of Bayesian optimal two-stage designs to the two-arm setting, compare \Cref{eq:minExpN}.
\end{itemize}
The remainder of this work develops exactly this extension, showing how the trinomial-tree branching concept and futility-erased partial contributions can be generalized to the two-arm matrix-search setting in order to obtain an optimal two-stage two-arm Bayes factor design. The most relevant application of such an extension is a phase II clinical trial with a treatment and control arm and a primary binary endpoint.

\section{Extending the Two-Arm Bayes Factor Design to a Two-Stage Setting}
\label{sec:method}
In this section, we outline how to extend the fixed-sample two-arm Bayes factor design for phase~II trials with binary endpoints detailed in \cite{kelterTwoArmTwoStage2026} to a two-stage design with a single interim analysis. The goal is to retain the simulation-free calibration philosophy: power and type-I-error are obtained by summing prior-predictive probabilities over suitable regions in the sample space, now taking into account the option to stop early for futility.

\subsection{Setup and notation}
\label{subsec:notationAndSetup}
We consider a two-arm trial with a control group ($j=1$) and a treatment group ($j=2$). Let $n_2^{(j)}$ denote the planned total sample size in arm $j$ at the final analysis, and let $n_1^{(j)} < n_2^{(j)}$ denote the sample size in arm $j$ at the interim analysis. For simplicity, we assume that the allocation ratio is fixed and that the interim occurs after $n_1^{(1)}$ and $n_1^{(2)}$ patients have been observed in the control and treatment arm, respectively. We denote by
\[
  X_j \sim \text{Bin}\bigl(n_1^{(j)}, p_j\bigr), 
  \qquad 
  Z_j \sim \text{Bin}\bigl(n_2^{(j)} - n_1^{(j)}, p_j\bigr),
  \qquad j=1,2,
\]
the numbers of successes in arm $j$ in the first and second stage, respectively. The final totals are
\[
  Y_j = X_j + Z_j, \qquad j = 1,2.
\]
Thus $(X_1,X_2)$ describes the interim data and $(Y_1,Y_2)$ the final data, if the trial is continued to the second stage.

As in the two-arm fixed-sample setting, we write $H_0$ and $H_1$ for the null and alternative hypotheses of interest. In a phase II trial, two hypotheses are typically of interest. One tests equality of response probabilities in treatment and control, $p_1 = p_2$, where the former receives the novel drug and the latter standard of care or placebo:
\begin{align}\label{eq:twoSided}
    H_0: p_1 = p_2 \quad \text{versus} \quad H_1: p_1 \neq p_2
\end{align}
A convenient reparameterization introduces the \emph{difference} $\eta = p_2 - p_1$ and the \emph{grand mean} $\zeta = \tfrac{1}{2}(p_1 + p_2)$, so that
\[
p_1 = \zeta - \frac{\eta}{2}, \quad p_2 = \zeta + \frac{\eta}{2},
\]
and the hypotheses become
\begin{align}\label{eq:twoSidedParameterized}
    H_0: \eta = 0 \quad \text{versus} \quad H_1: \eta \neq 0.
\end{align}
This parameterization originates from \cite{Gunel1974} and has been used in subsequent work, see also \cite{Dickey1970a}, \cite{Jamil2017} and \cite{Kelter2025}. In a phase~IIb setting, testing $H_0: \eta = 0$ is attractive because it allows explicit evidence \emph{for} equal efficacy of novel treatment and control. If evidence instead supports $H_1: \eta \neq 0$, either $p_1 > p_2$ (control more effective than treatment) or $p_1 < p_2$ (treatment more effective than control) may occur, and estimating $p_1$ and $p_2$ post-hoc should supplement this hypothesis test for a more complete interpretation.

An alternative is to use directional tests:
\begin{align}
    &H_0: \eta \leq 0 \quad \text{versus} \quad H_1: \eta > 0\label{eq:directionalParameterized}\\
    &H_0: \eta = 0 \quad \text{versus} \quad H_1: \eta > 0\label{eq:directionalParameterizedOneSided1}\\
    &H_0: \eta = 0 \quad \text{versus} \quad H_1: \eta < 0\label{eq:directionalParameterizedOneSided2}
\end{align}
In the first, $H_0$ states that placebo or standard of care is at least as effective as the novel treatment, while $H_1$ asserts superior efficacy of the novel treatment. The one-sided test with $H_0: \eta = 0$ versus $H_1: \eta > 0$ assumes $\eta < 0$ (i.e., $p_2 < p_1$) is a priori unrealistic, for example when the control arm receives standard of care and the treatment arm receives standard of care plus a non-interfering add-on, so the treatment success probability should be at least as large. Conversely, the one-sided test with $H_0: \eta = 0$ versus $H_1: \eta < 0$ is relevant when the endpoint measures failures; ruling out $\eta > 0$ a priori implies $p_2 \le p_1$, which is reasonable if the control group receives standard of care and the treatment group standard of care plus an add-on that cannot worsen outcomes.

\cite{kelterTwoArmTwoStage2026} uses conjugate beta design priors under each hypothesis, developed for the two-sided Bayes factor test
\[
  H_0: p_1 = p_2 \quad \text{versus} \quad H_1: p_1 \neq p_2.
\]
The beta--binomial model then yields closed-form prior-predictive probability mass functions
\begin{align}\label{eq:f1}
  f_1(x_1,x_2 \mid H_i)
  := \Pr(X_1 = x_1, X_2 = x_2 \mid H_i) \text{ (interim data under the two‑sided test)},
\end{align}
\begin{align}\label{eq:f2}
  f_2(z_1,z_2 \mid H_i)
  := \Pr(Z_1 = z_1, Z_2 = z_2 \mid H_i) \text{ (stage‑wise partition of the same data)},
\end{align}
and
\begin{align}\label{eq:f3}
  f(y_1,y_2 \mid H_i)
  := \Pr(Y_1 = y_1, Y_2 = y_2 \mid H_i) \text{ (final‑sample arm‑wise counts)},
\end{align}
for $i \in \{0,1\}$. Under $H_1: p_1 \neq p_2$, the arms receive independent beta design priors,
\[
  p_1 \mid H_1 \sim \Beta(\alpha_{1d}, \beta_{1d}), \quad
  p_2 \mid H_1 \sim \Beta(\alpha_{2d}, \beta_{2d}),
\]
and the prior-predictive probability mass function for interim data $(x_1,x_2)$ under the two‑sided test, with interim sample sizes $n_1^{(1)}$ and $n_1^{(2)}$, is
\[
  f_1(x_1,x_2 \mid H_1)
  = \binom{n_1^{(1)}}{x_1}
    \frac{B(\alpha_{1d} + x_1, \beta_{1d} + n_1^{(1)} - x_1)}
         {B(\alpha_{1d}, \beta_{1d})}
    \cdot
    \binom{n_1^{(2)}}{x_2}
    \frac{B(\alpha_{2d} + x_2, \beta_{2d} + n_1^{(2)} - x_2)}
         {B(\alpha_{2d}, \beta_{2d})},
\]
for $x_1 \in \{0,\dots,n_1^{(1)}\}$, $x_2 \in \{0,\dots,n_1^{(2)}\}$. Under the null hypothesis $H_0: p_1 = p_2 = p$, the common parameter $p$ follows a beta prior $\Beta(\alpha_{0d}, \beta_{0d})$, so the corresponding joint prior-predictive probability mass function for the arm-wise interim totals $(x_1,x_2)$ is
\[
  f_1(x_1,x_2 \mid H_0)
  = \binom{n_1^{(1)}}{x_1}
    \binom{n_1^{(2)}}{x_2}
    \frac{B(\alpha_{0d} + x_1 + x_2, \beta_{0d} + n_1^{(1)} + n_1^{(2)} - x_1 - x_2)}
         {B(\alpha_{0d}, \beta_{0d})},
\]
for $x_1 \in \{0,\dots,n_1^{(1)}\}$, $x_2 \in \{0,\dots,n_1^{(2)}\}$. This joint probability mass function governs the distribution of the interim data under $H_0$ and is used in the calculation of the type-I-error probabilities. For computational purposes, the Bayesian power and type-I-error probabilities are obtained by summing the prior-predictive probability mass function $f_1(x_1,x_2 \mid H_i)$ over all $(x_1,x_2)$ that fall into the respective critical regions.

The probability mass function
\[
  f(y_1,y_2 \mid H_i)
  := \Pr(Y_1 = y_1, Y_2 = y_2 \mid H_i),
\]
for the final-sample arm-wise counts is precisely the two-arm prior-predictive distribution already used in the fixed-sample Bayes factor sample size calculations for the two-sided test $H_0: p_1 = p_2$ versus $H_1: p_1 \neq p_2$, with final sample sizes $n_2^{(1)}$ and $n_2^{(2)}$ and, potentially, different design-prior parameters. The only difference compared to the interim probability mass function $f_1(x_1,x_2 \mid H_i)$ is that the counts $(y_1,y_2)$ now refer to the final totals and the corresponding final sample sizes, rather than the interim sizes. All these prior-predictive probability mass functions are later used to compute the Bayesian power and type-I-error probabilities under the respective design priors.

In closing this subsection, we note that we only detail the derivations for the two-sided test of $H_0:p_1=p_2$ versus $H_1:p_1\neq p_2$, but the prior-predictive probability mass functions and the corresponding Bayes factors for the above directional tests have been derived by \cite{kelterTwoArmTwoStage2026}. The general approach outlined for the two-sided test can therefore easily be extended to the directional tests, which are often more relevant in the context of a phase II trial. For the directional tests, the beta design and analysis priors change to truncated versions on the parameter spaces associated with $H_0$ and $H_1$, and the prior-predictive probability mass functions do likewise. For details, also on the resulting Bayes factors for these directional tests, we refer to the Appendix of \cite{kelterTwoArmTwoStage2026}.

\subsection{Bayes factors at interim and final analysis}

Let $BF_{01}(x_1,x_2)$ denote the Bayes factor in favour of $H_0$ based on the interim counts $(X_1,X_2) = (x_1,x_2)$, and $BF_{01}(y_1,y_2)$ the Bayes factor based on the final totals $(Y_1,Y_2) = (y_1,y_2)$. These are exactly the two-arm Bayes factors derived in the fixed-$n$ two-arm setting, now evaluated at the interim and final sample sizes. We choose two thresholds:
\begin{itemize}
  \item $k < 1$ for evidence against $H_0$ (efficacy boundary).
  \item $k_f > 1$ for evidence in favour of $H_0$ (futility boundary).
\end{itemize}
Using the $BF_{01}$ orientation, a small Bayes factor ($BF_{01} < k$ for e.g. $k=1/3$ or $k=1/10$) indicates evidence against $H_0$ in favour of $H_1$, while a large Bayes factor ($BF_{01} \ge k_f$) indicates evidence for $H_0$.\footnote{This is in line with the interpretation of p-values, simplifying the use and interpretation for frequentists.} At the interim analysis, we define:
\begin{align*}
  \mathcal{F}_1 
  &:= \bigl\{ (x_1,x_2) : BF_{01}(x_1,x_2) \geq k_f \bigr\} \quad \text{(the futility region)},\\
  \mathcal{E}_1 
  &:= \bigl\{ (x_1,x_2) : BF_{01}(x_1,x_2) < k \bigr\} \quad \text{(the interim efficacy region)},\\
  \mathcal{C}_1 
  &:= \bigl\{ (x_1,x_2) : (x_1,x_2) \notin \mathcal{F}_1 \cup \mathcal{E}_1 \bigr\}
  \quad \text{(the continuation region)}.
\end{align*}
The interim efficacy region is optional, as it only becomes relevant when stopping for efficacy after the interim analysis is allowed for. In our current design, this is not the case but we briefly explain how a possible extension could look like in Section~\ref{subsec:early-eff} below. At the final analysis, we define the final efficacy region
\[
  \mathcal{E}_2 
  := \bigl\{ (y_1,y_2) : BF_{01}(y_1,y_2) < k \bigr\}.
\]
We focus on the case where early stopping is allowed only for futility.

\subsection{Unadjusted (fixed-sample) operating characteristics}

Ignoring the interim look and treating the design as fixed-sample with total sample sizes $n_2^{(1)}$ and $n_2^{(2)}$, the Bayesian analogues of type-I-error rate and power are given by
\begin{align*}
  P_{\text{naive}}^{(0)} 
  &:= \Pr\bigl(BF_{01}(Y_1,Y_2) < k \mid H_0\bigr) 
  = \sum_{(y_1,y_2)\in\mathcal{E}_2} f(y_1,y_2 \mid H_0),\\
  P_{\text{naive}}^{(1)} 
  &:= \Pr\bigl(BF_{01}(Y_1,Y_2) < k \mid H_1\bigr) 
  = \sum_{(y_1,y_2)\in\mathcal{E}_2} f(y_1,y_2 \mid H_1).
\end{align*}
The quantity $P_{\text{naive}}^{(0)}$ is the naive Bayesian type-I-error rate, whereas $P_{\text{naive}}^{(1)}$ is the naive Bayesian power. These are the quantities used in the fixed-sample two-arm Bayes factor sample size calculations, compare \cite{PawelHeld2025}, \cite{KelterPawel2025,KelterPawelTwoStage2025} and \cite{kelterTwoArmTwoStage2026}.

However, once we introduce the option to stop early for futility at the interim analysis (i.e.\ whenever $(X_1,X_2)\in\mathcal{F}_1$), these unadjusted probabilities overestimate the true power and type-I-error rate of the resulting two-stage design, for the same reason as in the single-arm two-stage setting: some data trajectories that would have contributed to $P_{\text{naive}}^{(i)}$ are no longer possible because the trial would have been stopped early.

\subsection{Futility-erased partial power and type-I-error}

The key idea is to identify those trajectories that (i) would have stopped for futility at the interim analysis, but (ii) would have produced a final Bayes factor indicating efficacy if the trial had continued to the second stage. In the single-arm two-stage setting, these trajectories gave rise to the so-called ``futility-erased partial power'' and ``futility-erased partial type-I-error'', which must be subtracted from $P_{\text{naive}}^{(1)}$ and $P_{\text{naive}}^{(0)}$, respectively. We now derive the two-arm analogue.

\begin{lemma}\label{lemma:1}
    For a given hypothesis $H_i$ ($i \in \{0,1\}$), the joint prior-predictive distribution of $(X_1,X_2,Z_1,Z_2)$ factorizes as
\begin{align*}
  \Pr(X_1 = x_1, X_2 = x_2, Z_1 = z_1, Z_2 = z_2 \mid H_i)
  = f_1(x_1,x_2 \mid H_i) \, f_2(z_1,z_2 \mid H_i),
\end{align*}
where $f_1$ and $f_2$ are the stage-wise prior-predictive probability mass functions in \Cref{eq:f1} and \Cref{eq:f2}.
\end{lemma}
\begin{proof}
    See the Appendix.
\end{proof}
Based on \Cref{lemma:1}, one can derive the following two-arm analogue of the futility-erased partial contribution to Bayesian power or type-I-error rate for a two-arm sequential two-stage design with binary endpoints in both groups:
\begin{theorem}\label{theorem:1}
The two-arm futility-erased partial contribution $\Delta^{(i)}$ for hypothesis $H_i$, $i\in \{0,1\}$ to Bayesian power or type-I-error rate is given as follows:
\begin{align*}
  \Delta^{(i)} 
  &:= \Pr\bigl(
    \underbrace{BF_{01}(Y_1,Y_2)< k}_{\text{reach efficacy based on final data}},
    \ \underbrace{BF_{01}(X_1,X_2)\geq k_f}_{\text{reach futility based on interim data}}
    \mid H_i
  \bigr)\\
  &= \sum_{(x_1,x_2)\in\mathcal{F}_1}
     \sum_{\substack{(z_1,z_2):\\(x_1+z_1,x_2+z_2)\in\mathcal{E}_2}}
       f_1(x_1,x_2 \mid H_i)\, f_2(z_1,z_2 \mid H_i).
\end{align*}
\end{theorem}
\begin{proof}
    See the Appendix.
\end{proof}
In \Cref{theorem:1}, the inner sum runs over all second-stage increments $(z_1,z_2)$ that, together with a futility-interim pair $(x_1,x_2)\in\mathcal{F}_1$, would have led to a final total $(y_1,y_2) = (x_1+z_1,x_2+z_2)$ in the final efficacy region $\mathcal{E}_2$. In the context of \Cref{fig:twoStageDesign}, $\Delta^{(i)}$ corresponds to the probability of the lowest of the three blue trajectories. If $H_i=H_0$, it is the futility-erased partial contribution to the Bayesian type-I-error rate of the resulting two-stage design. If $H_i=H_1$, it is the futility-erased partial contribution to the Bayesian power of the resulting two-stage design.

Intuitively, $\Delta^{(i)}$ is a two-arm, four-dimensional version of the ``futility-erased partial power'' described in the trinomial-tree framework of \cite{KelterPawelTwoStage2025} for the single-arm two-stage design. It collects exactly those trajectories that are counted in the fixed-sample power or type-I-error, but are no longer reachable when the trial is stopped for futility after the interim analysis.

In the Appendix, we provide another version of \Cref{lemma:1} and \Cref{theorem:1} which prove the factorization and double-sum expression also for the directional tests given in \Cref{eq:directionalParameterized} to \Cref{eq:directionalParameterizedOneSided2}.

\subsection{Corrected operating characteristics for the two-stage two-arm design}

The corrected Bayesian type-I-error rate and power of the two-stage design with a single futility interim are obtained by subtracting the futility-erased partial contributions from the naive fixed-sample probabilities:
\begin{align*}
  P_{\text{corr}}^{(0)}
  &:= \Pr\bigl(\text{declare efficacy at final} \mid H_0\bigr)\\
  &= P_{\text{naive}}^{(0)} - \Delta^{(0)},\\[1ex]
  P_{\text{corr}}^{(1)}
  &:= \Pr\bigl(\text{declare efficacy at final} \mid H_1\bigr)\\
  & =P_{\text{naive}}^{(1)} - \Delta^{(1)}.
\end{align*}
Here, $P_{\text{corr}}^{(0)}$ is the corrected Bayesian type-I-error rate, whereas $P_{\text{corr}}^{(1)}$ is the corrected Bayesian power. In the simplest case with early stopping only for futility (no early efficacy stopping), the second equation reduces to
\[
  P_{\text{corr}}^{(1)} = P_{\text{naive}}^{(1)} - \Delta^{(1)},
\]
because all trajectories contributing to efficacy must go through the final analysis and those that would have stopped for futility but later yielded efficacy are precisely the ones counted in $\Delta^{(1)}$.

If early stopping for efficacy at the interim analysis is allowed (see \Cref{subsec:early-eff} below), the expression for $P_{\text{corr}}^{(1)}$ includes an additional term for interim efficacy.

\subsection{Extension to early stopping for efficacy}
\label{subsec:early-eff}

If early stopping for efficacy at the interim analysis is allowed, $$P_{\text{corr}}^{(1)}=\text{Pr}(\text{declare efficacy at interim or final}|H_1)$$
and the decision rule gains a second type of early stop: in addition to futility, the trial may stop when $(X_1,X_2)\in\mathcal{E}_1^{(1)}$ with $BF_{01}(x_1,x_2)<k$. In this case, the overall Bayesian power under $H_1$ decomposes into
\[
P_{\text{corr}}^{(1)}
= \underbrace{
    \sum_{(x_1,x_2)\in\mathcal{E}_1^{(1)}} f_1(x_1,x_2 \mid H_1)
  }_{\text{efficacy at interim}}
  + 
  \Bigl[
    P_{\text{naive}}^{(1)}
    - \Delta^{(1)}
    - \Omega^{(1)}
  \Bigr],
\]
where $P_{\text{naive}}^{(1)}$ is the Bayesian power in the fixed-sample design, $\Delta^{(1)}$ is the futility-erased partial power (paths in the futility region $\mathcal{F}_1^{(1)}$ that would have fallen into the final efficacy region $\mathcal{E}_2^{(1)}$ had the trial continued, corresponding to the lowest blue trajectory in \Cref{fig:twoStageDesign}), and $\Omega^{(1)}$ is the analogous efficacy-erased partial power (paths in the efficacy region $\mathcal{E}_1^{(1)}$ that would have also satisfied the final efficacy rule but are now counted only once, in the first summand. As they are included both in the first summand and $P_{\text{naive}}^{(1)}$, they must be subtracted once. The corresponding path is the upper blue trajectory in \Cref{fig:twoStageDesign}.). The first term accounts for trajectories stopped for efficacy at the interim, while the bracketed term corresponds to the probability of declaring efficacy at the final analysis, after correcting for both erased trajectories which are ``cut off'' because one stops for futility or efficacy. Under $H_0$, analogous decompositions can be derived for the type-I-error rate, and the calibration algorithm outlined in the following subsection could proceed along the same lines, with the additional constraint that early efficacy contributions must be included in the power and type-I-error targets. In this paper, we solely consider stopping early for futility, but future research could deal with extensions involving designs which allow early stopping for efficacy.

\subsection{Corrections for the probability of compelling evidence for the two-stage two-arm design}
Next to the power and type-I-error, another operating characteristic of the trial design which changes when introducing an interim analysis is the probability of compelling evidence
$$\Pr(\mathrm{CE}_{\mathrm{fix}}\mid H_0)=P(BF_{01}(X_1,X_2)\geq k|H_0)$$
where the calibration requires the latter to achieve at least a minimum probability $f\in (0,1)$:
$$P(BF_{01}(X_1,X_2)\geq k|H_0)>f.$$
In principle, for a given final sample size $(n_2^{(1)},n_2^{(2)})$ the corrected two-stage probability of obtaining compelling evidence for $H_0$, denoted $\Pr(\mathrm{CE}_{2\mathrm{st}}\mid H_0)$, may exceed its fixed-sample counterpart $\Pr(\mathrm{CE}_{\mathrm{fix}}\mid H_0)$ associated with the same totals. The reason is that, in the two-stage design, interim outcomes that fall into the futility region are counted immediately as compelling evidence for $H_0$, whereas in the corresponding fixed-sample design the trial would necessarily continue to the final analysis before $\mathrm{CE}$ is assessed. As shown in the appendix, $\Pr(\mathrm{CE}_{2\mathrm{st}}\mid H_0)$ therefore decomposes into the sum of the probability of early futility stopping and the probability of reaching compelling evidence for $H_0$ at the final analysis after continuation, which implies that $\Pr(\mathrm{CE}_{2\mathrm{st}}\mid H_0)\ge \Pr(\mathrm{CE}_{\mathrm{fix}}\mid H_0)$, with strict inequality whenever some interim futility outcomes would not lead to compelling evidence for $H_0$ in the fixed-sample design.

In the calibration algorithm described in the following section, the constraint on $\Pr(\mathrm{CE}\mid H_0)$ is nevertheless enforced already at the fixed-sample level in Step~1. This should be viewed as a conservative feasibility screen: by requiring the fixed-sample design to attain the desired probability of compelling evidence for $H_0$, Step~1 tends to exclude unrealistically small final sample sizes for which even the non-sequential fixed-sample procedure cannot meet the evidence requirement. In Step~2 of the calibration algorithm detailed in the following section, the constraint is then checked again using the corrected two-stage quantity $\Pr(\mathrm{CE}_{2\mathrm{st}}\mid H_0)$ for each candidate interim design. Thus, Step~1 of the calibration algorithm detailed next provides a conservative screening device, while the final calibration of the sequential design is based on the corrected two-stage operating characteristics derived in Appendix \ref{sec:corrections}.

\section{Calibration algorithm in the two-arm two-stage setting}
\label{subsec:calibration}

The calibration problem now is to choose interim and final sample sizes $(n_1^{(1)},n_1^{(2)},n_2^{(1)},n_2^{(2)})$ such that the corrected operating characteristics satisfy prespecified bounds, e.g.
\begin{align}\label{eq:target_constraints_calibration_t1e_and_power}
  P_{\text{corr}}^{(0)} \le \alpha, 
  \qquad 
  P_{\text{corr}}^{(1)} \ge 1-\beta,
\end{align}
and optionally a constraint on the probability to stop for futility under $H_0$,
\begin{align}\label{eq:target_constraints_calibration_pceH0}
  \Pr\bigl( BF_{01}(X_1,X_2) \ge k_f \mid H_0 \bigr) \ge f.
\end{align}
The probability to stop for futility under $H_0$ above
quantifies the chance that the interim data provide \emph{compelling evidence for} $H_0$ and the trial is stopped early. A closely related idea appears in the Bayesian reanalysis of null results by \cite{hoekstraBayesianReanalysisNull2018}, who refer to \emph{``compelling evidence for the null hypothesis''} when Bayes factors $BF_{01}$ exceed a fixed threshold. In the context of Bayesian group‑sequential and phase II designs, early stopping rules for futility are studied using Bayesian posterior or predictive probabilities; examples and discussions of such rules and their operating characteristics can be found in \cite{HeathEtAl2020} and in the tutorial on modern Bayesian methods in clinical trials by \cite{MuehlemannEtAl2023}, as well as \cite{jiangComparingBayesianEarly2020}. In what follows, we use the term \textit{compelling evidence for $H_0$}, in line with \cite{KelterPawel2025}, \cite{KelterPawelTwoStage2025} and \cite{kelterTwoArmTwoStage2026}.

A natural calibration algorithm based on \Cref{lemma:1} and \Cref{theorem:1}, which mirrors the single‑arm two‑stage design, is now given as follows.

\begin{enumerate}
  \item \textbf{Step 1 (fixed-sample calibration):} For a fixed allocation ratio, search over a grid of total sample sizes $(n_2^{(1)},n_2^{(2)})$ with $n_{\text{min}}\leq n_2^{(j)}\leq n_{\text{max}}$, $j=1,2$ for a realistic range of minimum and maximum sample sizes $n_{\text{min}}$ and $n_{\text{max}}$, until the naive fixed-sample Bayesian power $P_{\text{naive}}^{(1)}$ exceeds a prespecified target $1-\beta$ (e.g. 80\%), possibly with a small cushion (e.g., $1-\beta + \delta$ for a small $\delta>0$). This yields a ``sufficient'' full-trial size.
  \item \textbf{Step 2 (two-stage calibration):} For the chosen full-trial size, search over interim sample sizes $(n_1^{(1)},n_1^{(2)})$ and identify those designs that satisfy the constraints on the corrected operating characteristics $P_{\text{corr}}^{(0)}$ and $P_{\text{corr}}^{(1)}$ in \Cref{eq:target_constraints_calibration_t1e_and_power}, as well as any optional constraint on the probability of early stopping for futility under $H_0$ in \Cref{eq:target_constraints_calibration_pceH0}. Among these, select the design that optimizes the desired criterion, e.g., minimizes the expected total sample size under $H_0$.
\end{enumerate}
\begin{figure}[h!]
    \centering
    \includegraphics[width=1\linewidth]{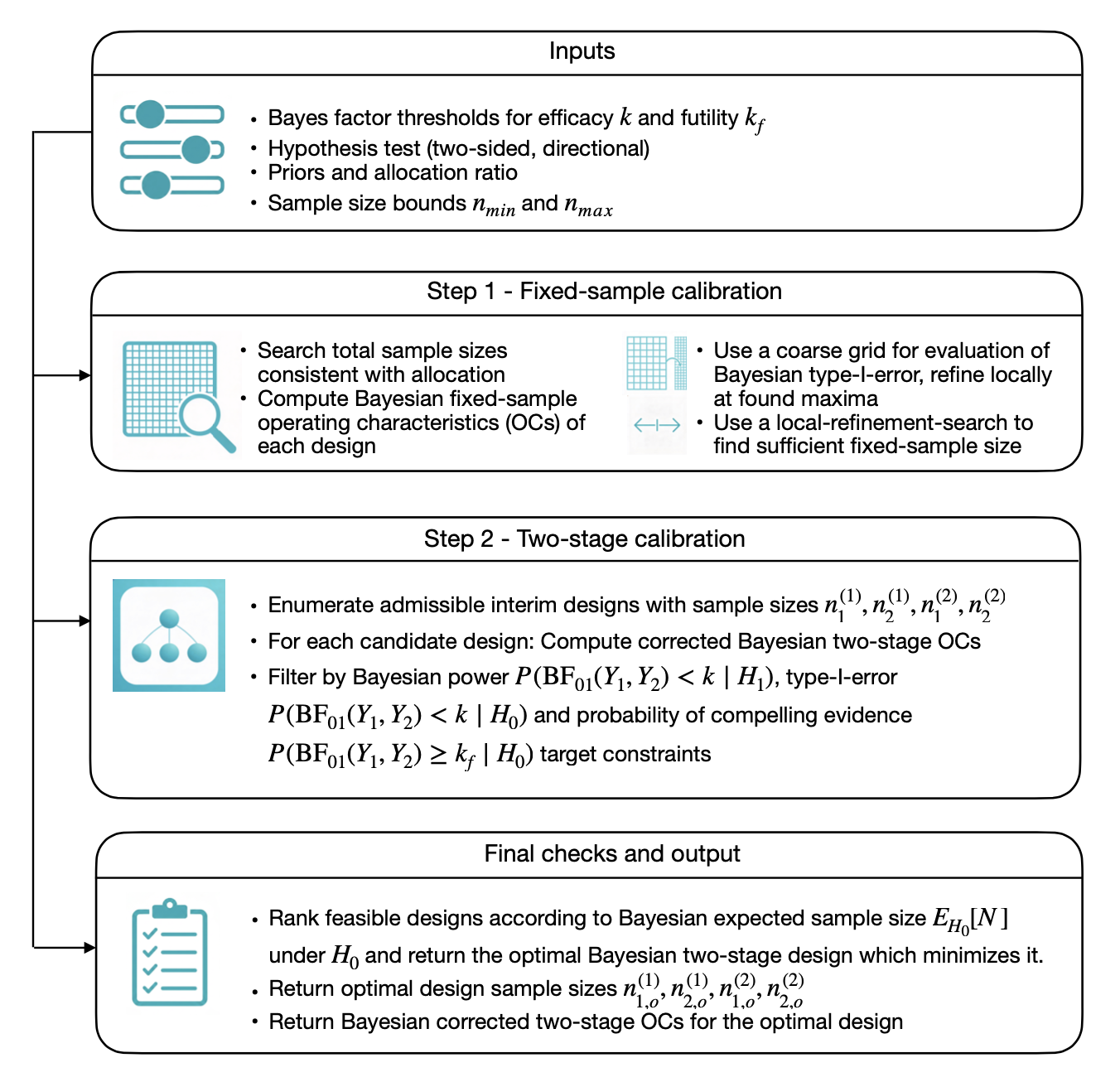}
    \caption{Illustration of the calibration algorithm for finding an optimal Bayesian two-arm two-stage design. Note that neither the use of Bayes factors nor binary endpoints in both trial arms are required for the algorithm to work.}
    \label{fig:algorithm_bayes}
\end{figure}
\Cref{fig:algorithm_bayes} visualizes the calibration algorithm for finding an optimal Bayesian two-arm two-stage design. Note, that nothing requires to use Bayes factors after all. One could also use posterior probabilities or any other test statistic and proceed likewise, replacing the computation of Bayes factors with the other test statistic of choice. 

Now, if no such design is found, the full-trial size can be increased and the process repeated. The resulting algorithm in detail then looks as follows.

\begin{enumerate}
\item \textbf{Input:} Design and analysis priors under $H_0$ and $H_1$, Bayes‑factor thresholds $k<11$ and $k_f>1$, target error bounds $\alpha,\beta\in(0,1)$, allocation ratio, hypotheses to test and (optionally) a target probability of early stopping for futility $f\in(0,1)$ under $H_0$.

\item \textbf{Find a sufficiently large fixed‑sample size:} For an increasing grid of total sample sizes $(n_2^{(1)},n_2^{(2)})$ with $n_{\text{min}}\leq n_2^{(j)}\leq n_{\text{max}}$ for $j=1,2$ (e.g.\ constrained by a fixed allocation ratio), compute the naive fixed‑sample Bayesian power
  \[
    P_{\text{naive}}^{(1)} = \sum_{(y_1,y_2)\in\mathcal{E}_2} f(y_1,y_2\mid H_1)
  \]
  until it exceeds the target $1-\beta$, possibly with a small cushion (e.g., up to $1-\beta + \delta$ for a small $\delta>0$).\footnote{This step ensures that the resulting Bayesian power target can be reached in principle. Based on \Cref{theorem:1}, the power in the two‑stage design can only decrease when an interim analysis that allows stopping for futility is introduced. As a consequence, if the fixed‑sample design cannot reach the target power, no two‑stage design with that or smaller total sample size can. This computation is performed by means of the matrix‑search algorithm outlined in \Cref{subsec:matrix-search}.} Optionally, compute the naive fixed-sample probability of compelling evidence $\text{Pr(CE}_{\text{fix}}|H_0)$ and calibrate it according to \Cref{eq:target_constraints_calibration_pceH0}.\footnote{Note that for screening for a sufficient fixed-sample size, we do not use $\text{ CE}_{H_0}^{\text{2st}}$ as defined in \Cref{eq:pce2st}, which is the corrected two-stage probability of compelling evidence. We solely use the fixed-sample probability of compelling evidence, as the corrected two-stage probability of compelling evidence must increase when introducing an interim analysis which allows stopping for futility. See Appendix \Cref{sec:correctionFutilityStopping}.} Let $(n_2^{(1)},n_2^{(2)})$ denote a candidate full‑trial size that achieves this (or lies in a small surrounding region).

\item \textbf{Compute fixed‑sample type‑I‑error at that size:} For the chosen total sample size $(n_2^{(1)},n_2^{(2)})$, compute the corresponding naive fixed‑sample type‑I‑error
  \[
    P_{\text{naive}}^{(0)} = \sum_{(y_1,y_2)\in\mathcal{E}_2} f(y_1,y_2\mid H_0),
  \]
  where $\mathcal{E}_2 := \{(y_1,y_2) : BF_{01}(y_1,y_2) < k\}$.\footnote{Note that the type-I-error rate also can only decrease when introducing an interim analysis which allows to stop for futility only. As a consequence, even if the resulting type-I-error rate for that sample size $(n_2^{(1)},n_2^{(2)})$ does not meet the desired requirements, the resulting two-stage design can still meet those.}

\item \textbf{Iterate over interim sample sizes:} For the fixed final sample size pair $(n_2^{(1)},n_2^{(2)})$, consider candidate interim sample size pairs $(n_1^{(1)},n_1^{(2)})$\footnote{For example, as fractions of $n_2^{(1)},n_2^{(2)}$, or starting from a small value such as $n_1^{(1)}=n_1^{(2)}=5$, iterating up to $n_1^{(1)}=n_1^{(2)}=n_2^{(1)}-1$ for balanced randomization, and analogously for non‑balanced randomization.}:

  \begin{enumerate}
  \item Compute the stage‑wise prior‑predictive probability mass functions $f_1(\cdot,\cdot \mid H_i)$ and $f_2(\cdot,\cdot \mid H_i)$ for $i\in\{0,1\}$.\footnote{The prior‑predictive probability mass functions for the two‑sided and directional tests are available in the Appendix of \cite{kelterTwoArmTwoStage2026}. For the two-sided test, see also \Cref{eq:f1} and \Cref{eq:f2}.}

  \item Identify the interim futility region $\mathcal{F}_1^{(i)}$ and the final efficacy region $\mathcal{E}_2^{(i)}$ by evaluating the Bayes factor at all possible combinations of $(x_1,x_2)$ and $(y_1,y_2)$.\footnote{All Bayes factors for the two‑sided and directional tests are available in the Appendix of \cite{kelterTwoArmTwoStage2026} and can be computed via standard numerical integration.}

  \item Compute the futility‑erased partial contributions
    \[
      \Delta^{(i)}
      = \sum_{(x_1,x_2)\in\mathcal{F}_1^{(i)}}
        \sum_{\substack{(z_1,z_2):\\(x_1+z_1,x_2+z_2)\in\mathcal{E}_2^{(i)}}}
          f_1(x_1,x_2 \mid H_i)\, f_2(z_1,z_2 \mid H_i),
        \qquad i\in\{0,1\},
    \]
    as provided in \Cref{theorem:1}.

  \item Obtain the corrected operating characteristics
    \[
      P_{\text{corr}}^{(0)} = P_{\text{naive}}^{(0)} - \Delta^{(0)},
      \quad
      P_{\text{corr}}^{(1)} = P_{\text{naive}}^{(1)} - \Delta^{(1)}.
    \]

  \item Compute the probability of early stopping for futility under $H_0$:
    \[
      \Pr(\text{futility at interim} \mid H_0)
      = \sum_{(x_1,x_2)\in\mathcal{F}_1^{(0)}} f_1(x_1,x_2 \mid H_0),
    \]
    and the expected total sample size under $H_0$:
    \[
      E_{H_0}[N]
      = N_1 \Pr(\text{futility at interim}\mid H_0)
        + N_2 \bigl(1 - \Pr(\text{futility at interim}\mid H_0)\bigr),
    \]
    where $N_1 = n_1^{(1)} + n_1^{(2)}$ and $N_2 = n_2^{(1)} + n_2^{(2)}$.
  \end{enumerate}

\item \textbf{Design selection:} Among all interim‑size pairs $(n_1^{(1)},n_1^{(2)})$ for the fixed total size $(n_2^{(1)},n_2^{(2)})$ that satisfy the constraints
  \[
    P_{\text{corr}}^{(0)} \le \alpha, \quad
    P_{\text{corr}}^{(1)} \ge 1 - \beta,
  \]
  and (optionally) $\text{ CE}_{H_0}^{\text{2st}}>f$ with $\text{ CE}_{H_0}^{\text{2st}}$ as defined in \Cref{eq:pce2st}, choose the one that minimizes $E_{H_0}[N]$ (or some other desired criterion) and call it the \textit{Bayesian optimal two‑stage design}, in line with the notation in \cite{KelterPawelTwoStage2025}.

\item \textbf{If no such design exists:} Increase the final sample size $(n_2^{(1)},n_2^{(2)})$ and repeat steps 2–5 until a design is found that satisfies all constraints.
\end{enumerate}

In line with the fifth step of selecting an optimal design among all trial designs which fulfill the required conditions on Bayesian type-I-error rate and power, we formally define the \textit{optimal two-arm two-stage Bayes factor design (for binary endpoints)} as follows:
\begin{definition}[Optimal two-arm two-stage Bayes factor design for binary endpoints]
Let $\alpha \in (0,1)$ and $\beta \in (0,1)$ be given, and let
$n_{1,\min}^{(1)}, n_{1,\min}^{(2)}$ denote the minimum interim sample sizes
at which the trial may stop for futility, and let
$n_{2,\max}^{(1)}, n_{2,\max}^{(2)}$ denote the maximum final sample sizes.
For a given Bayes-factor threshold $k<1$, the optimal two-arm two-stage
Bayes factor design for binary endpoints is any admissible design
$(n_1^{(1)},n_1^{(2)},n_2^{(1)},n_2^{(2)})$ that solves
\[
\min_{n_1^{(1)},n_1^{(2)},n_2^{(1)},n_2^{(2)}} E_{H_0}[N]
\]
subject to
\[
P_{\mathrm{corr}}^{(0)} \le \alpha, \qquad
P_{\mathrm{corr}}^{(1)} \ge 1-\beta, \qquad
\text{ and, optionally, }\text{ CE}_{H_0}^{2st}>f
\]
where $\text{CE}_{H_0}^{2st}$ is defined in \Cref{eq:pce2st}, and
\[
n_{1,\min}^{(j)} \le n_1^{(j)} < n_2^{(j)} \le n_{2,\max}^{(j)}, \qquad j=1,2.
\]
Here, $P_{\mathrm{corr}}^{(0)}$ and $P_{\mathrm{corr}}^{(1)}$
denote the corrected Bayesian type-I error and power of the two-stage design, and $\text{ CE}_{H_0}^{\text{2st}}$ the corrected probability of compelling evidence of the two-stage design.
\end{definition}
The algorithm is implemented in the R package \texttt{bfbin2arm}, which is available on CRAN.\footnote{The package and various vignettes illustrating the use are available under \url{https://cran.r-project.org/web/packages/bfbin2arm/index.html}.}


\section{Examples}
\label{sec:examples}

\subsection{Re-analysis of the Riociguat phase II trial in systemic sclerosis}
\label{subsec:riociguat}

To illustrate the proposed methodology, we reconsider the riociguat phase II trial in systemic sclerosis discussed in the fixed-sample two-arm Bayes factor design setting by \cite{kelterTwoArmTwoStage2026}, compare also \cite{khannaRiociguatPatientsEarly2020}.\footnote{A detailed software vignette including all relevant R code to recreate this example which also includes further explanations is available on CRAN.} The example is attractive because it represents a realistic two-arm phase II setting with a binary endpoint, while also showing that the practical behaviour of the proposed two-stage calibration algorithm depends strongly on the prior-predictive separation of the competing hypotheses.

Let \(p_1\) denote the response probability in the control arm and \(p_2\) the response probability in the treatment arm. In the riociguat example, the observed response rates are
\[
\hat p_1 = \frac{38}{22+38} \approx 0.6333,
\qquad
\hat p_2 = \frac{48}{48+11} \approx 0.8135.
\]
Since the observed response rate is higher in the treatment arm, we consider the one-sided superiority setting
\[
H_0: p_1 = p_2
\qquad \text{versus} \qquad
H_1: p_1 < p_2,
\]
implemented through the Bayes factor \(BF_{0+}\), that is, evidence against \(H_0\) corresponds to small values of the Bayes factor in favour of the null.

Throughout this example, the efficacy and futility thresholds are chosen as
\[
k = \frac{1}{10},
\qquad
k_f = 3.
\]
Thus, efficacy is declared when the Bayes factor falls below \(1/10\), whereas compelling evidence in favour of the null hypothesis is declared when the Bayes factor is at least \(3\). We calibrate the design to satisfy a Bayesian type-I-error bound of \(\alpha = 0.025\), Bayesian power \(1-\beta = 0.80\), and a lower bound of \(0.60\) on the probability of compelling evidence for \(H_0\).

\paragraph{Priors.}
We distinguish between design priors, which determine the prior-predictive operating characteristics used during calibration, and analysis priors, which enter the Bayes factor itself. For the riociguat example, we use a flat design prior under \(H_0\),
\[
p \mid H_0 \sim \mathrm{Beta}(1,1),
\]
and slightly informative design priors under \(H_1\),
\[
p_1 \mid H_1 \sim \mathrm{Beta}(1,3),
\qquad
p_2 \mid H_1 \sim \mathrm{Beta}(3,1).
\]
The corresponding analysis priors are chosen to be flat,
\[
p \mid H_0 \sim \mathrm{Beta}(1,1),
\qquad
p_1 \mid H_1 \sim \mathrm{Beta}(1,1),
\qquad
p_2 \mid H_1 \sim \mathrm{Beta}(1,1).
\]
This choice reflects the intended separation between planning and analysis: prior information is allowed to influence the calibration of the design, but the eventual evidential assessment through the Bayes factor is based on neutral analysis priors.

\begin{figure}[h!]
    \centering
    \includegraphics[width=0.88\textwidth]{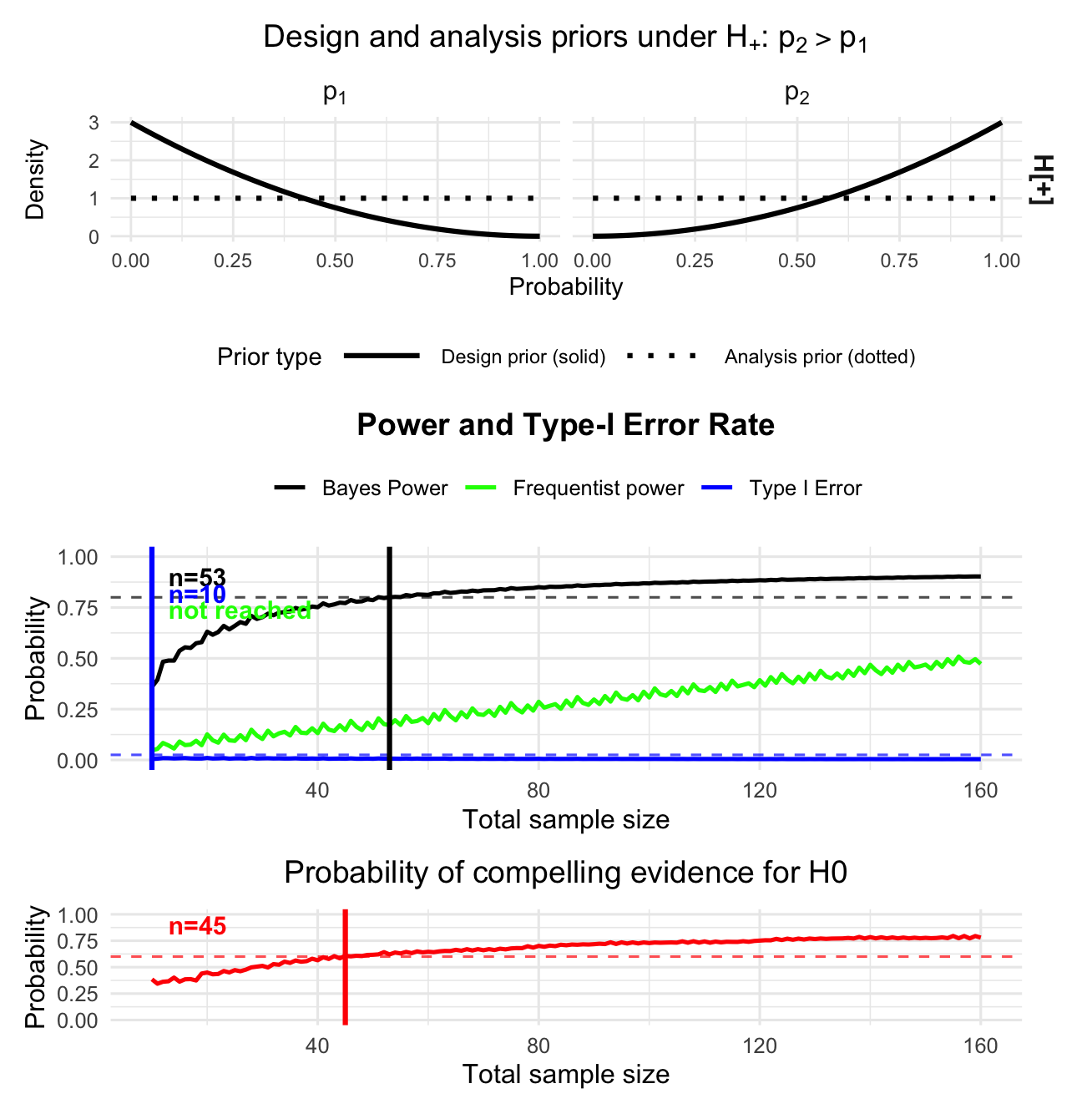}
    \caption{Calibrated one-stage Bayes factor design for the riociguat example. The figure illustrates the operating characteristics of the fixed-sample reference design under the chosen evidence thresholds and prior specification.}
    \label{fig:riociguat_onestage}
\end{figure}

\paragraph{One-stage reference design.}
Before constructing the two-stage design, it is helpful to inspect the corresponding fixed-sample reference design obtained under the same calibration targets. \Cref{fig:riociguat_onestage} shows the one-stage design calibrated to 80\% Bayesian power, 2.5\% Bayesian type-I error, and 60\% probability of compelling evidence, and requires \(N=53\) patients in total, corresponding to approximately to 27 patients per arm. At this sample size, the Bayesian power is about \(0.80\), the Bayesian type-I error is about \(0.007\), and the probability of compelling evidence in favour of \(H_0\) is about \(0.62\).

The corresponding one-stage design is useful as a benchmark, but it does not permit early stopping for futility. The practical question is therefore whether an interim analysis can be introduced without materially damaging the operating characteristics.

\paragraph{Two-stage design with mildly informative design priors.}
We now apply the proposed optimal two-stage calibration algorithm. The search is carried out under balanced randomization, with minimum interim sample sizes of 10 patients per arm and maximum final sample sizes of 80 patients per arm. Under the prior specification given above, the fixed-sample calibration step identifies a sufficient one-stage anchor with
\[
n_{21}=n_{22}=34.
\]
Conditional on this anchor, the second step of the algorithm searches over admissible interim sample sizes and selects the design minimizing the expected total sample size under \(H_0\).

\begin{figure}[h!]
    \centering
    \includegraphics[width=1\textwidth]{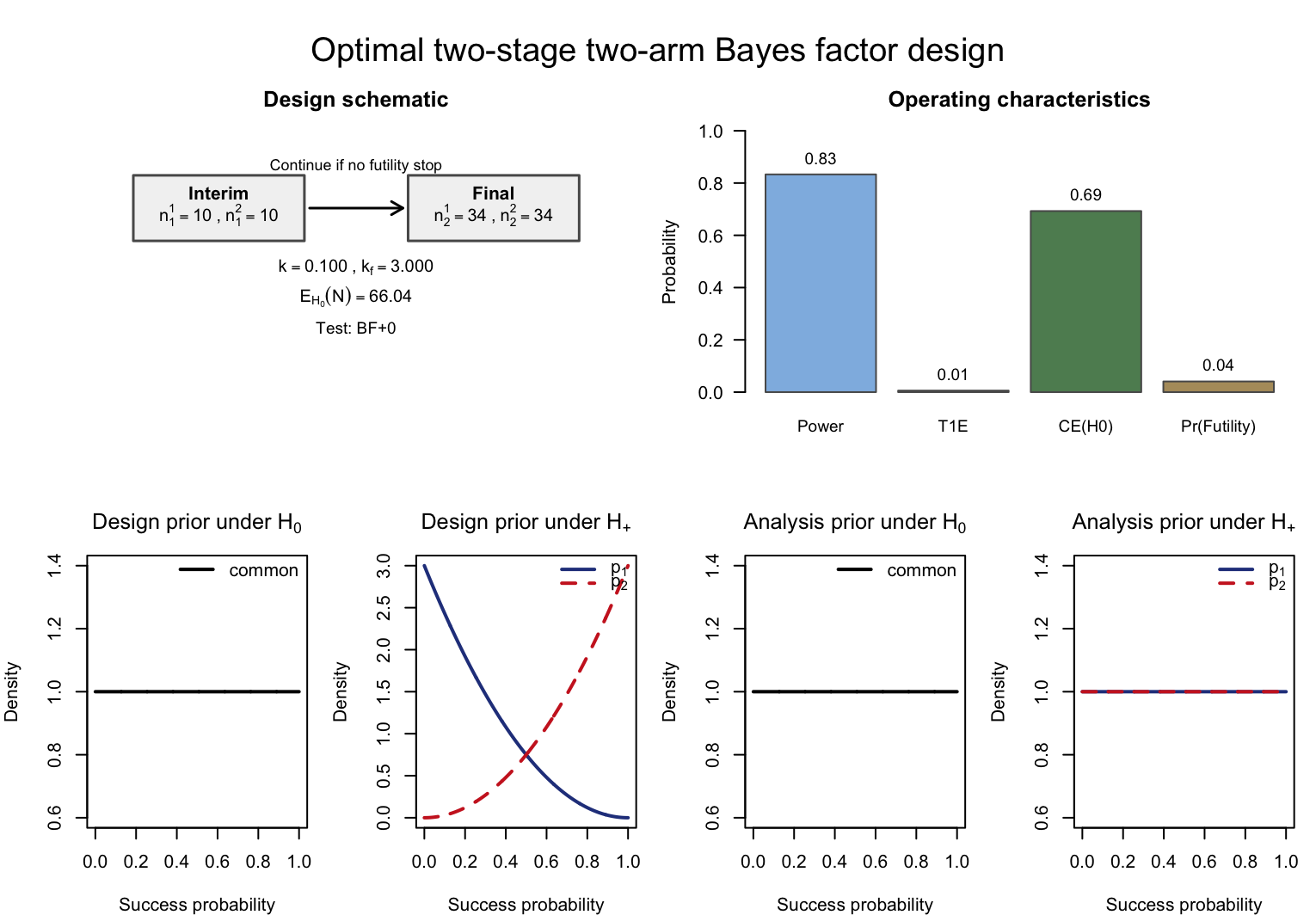}
    \caption{Optimal two-stage Bayes factor design for the riociguat example using mildly informative design priors under \(H_1\) and flat analysis priors. The figure visualizes the calibrated design and the corresponding prior specification used in planning.}
    \label{fig:riociguat_twostage_lessinformative}
\end{figure}

\Cref{fig:riociguat_twostage_lessinformative} shows the results of the calibration algorithm.\footnote{See also the software vignette available at \url{https://cran.r-project.org/web/packages/bfbin2arm/vignettes/bfbin2arm-twostage.html} for details on how to recreate the results and plots with the \texttt{bfbin2arm} R package.} For the riociguat example, the resulting optimal two-stage design is
\[
(n_{11}, n_{12}, n_{21}, n_{22}) = (10,10,34,34).
\]
Hence, the interim analysis is conducted after \(20\) patients in total, and the maximal sample size is \(68\). The corrected operating characteristics of this design are
\[
\text{Power} \approx 0.833,\qquad
\text{Type-I error} \approx 0.0058,\qquad
\mathrm{CE}_{H_0} \approx 0.693,
\]
with early stopping for futility under \(H_0\) occurring with probability about \(0.04\). The corresponding expected sample size under \(H_0\) is
\[
E_{H_0}[N] \approx 66.04.
\]

Several aspects are noteworthy. First, the corrected Bayesian power and type-I error remain comfortably within the desired design targets. Second, the price of allowing early futility stopping is small in terms of maximal sample size: the two-stage design increases the maximal sample size from \(53\) in the one-stage reference design to \(68\), but preserves the intended operating characteristics. Third, the actual gain in expected sample size under \(H_0\) is modest, because the futility stopping probability is only around \(4\%\). Thus, in this specific example, the interim analysis is feasible and principled, but it does not lead to a dramatic efficiency gain by itself.

\paragraph{Interpretation.}
The modest reduction in \(E_{H_0}[N]\) is not a deficiency of the algorithm. Rather, it is a consequence of the joint calibration constraints. The efficacy threshold \(k=1/10\) is fairly stringent, the futility threshold \(k_f=3\) requires non-trivial evidence in favour of \(H_0\), and the additional requirement \(\mathrm{CE}_{H_0}^{\text{2st}}\ge 0.60\) limits how aggressively null trajectories can be truncated at the interim analysis. Under such constraints, only a relatively small subset of null trajectories can be stopped early without compromising power or the evidence requirement for the null.

This example therefore illustrates an important practical point. A two-stage design does not automatically imply a substantially smaller expected sample size under \(H_0\). If the design priors under \(H_0\) and \(H_1\) are only moderately separated and the evidential thresholds are strict, then the calibrated futility rule may have only a limited opportunity to remove null trajectories early.

\paragraph{Two-stage design with more informative design priors.}
To investigate the effect of stronger prior-predictive separation, we keep the analysis priors and Bayes factor thresholds unchanged, but replace the design priors under \(H_1\) by the more informative specification
\[
p_1 \mid H_1 \sim \mathrm{Beta}(1,5),
\qquad
p_2 \mid H_1 \sim \mathrm{Beta}(5,1).
\]
This modification leaves the eventual Bayes factor analysis unchanged; it only affects the calibration stage by expressing a more concentrated prior expectation that the control arm has relatively low response probability and the treatment arm relatively high response probability.

\begin{figure}[h!]
    \centering
    \includegraphics[width=1\textwidth]{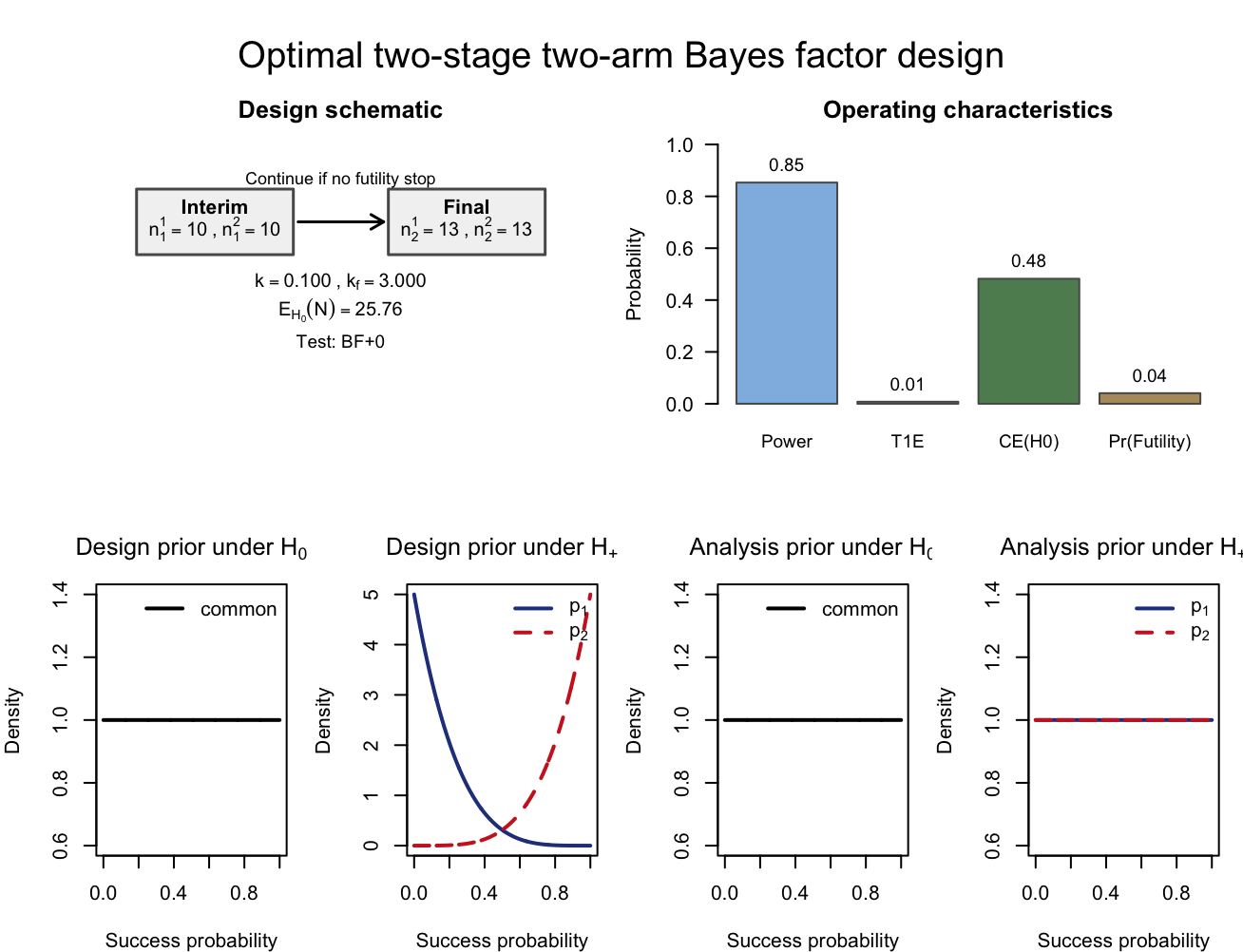}
    \caption{Optimal two-stage Bayes factor design for the riociguat example using more informative design priors under \(H_1\) and flat analysis priors. The figure visualizes the calibrated design and the corresponding prior specification used in planning.}
    \label{fig:two-stage-modified_more_informative_design_priors}
\end{figure}
\Cref{fig:two-stage-modified_more_informative_design_priors} shows the results.
Under this more informative design prior, the fixed-sample calibration step now identifies a sufficient one-stage anchor with
\[
n_{21}=n_{22}=13,
\]
so that the corresponding maximal total sample size is \(26\). The optimal two-stage design becomes
\[
(n_{11}, n_{12}, n_{21}, n_{22}) = (10,10,13,13).
\]
Its corrected operating characteristics are approximately
\[
\text{Power} \approx 0.853,\qquad
\text{Type-I error} \approx 0.0079,\qquad
\mathrm{CE}_{H_0} \approx 0.482,
\]
with futility stopping probability under \(H_0\) again close to \(0.04\), and expected sample size
\[
E_{H_0}[N] \approx 25.76.
\]

The striking feature is that the futility stopping probability changes very little, whereas the expected sample size under \(H_0\) decreases substantially, from about \(66.04\) to \(25.76\). The reason is that the more informative design priors separate the predictive distributions under \(H_0\) and \(H_1\) much more clearly. As a result, the fixed-sample anchor from Step 1 becomes much smaller, and the resulting two-stage design inherits this smaller overall scale. In other words, in this class of examples, the main efficiency gain is not driven by making the interim rule dramatically more aggressive, but by reducing the total sample size needed for the competing hypotheses to become distinguishable under the design priors.

\paragraph{Comparison of the three designs.}
Table~\ref{tab:riociguat_summary} summarizes the key operating characteristics of the one-stage reference design and the two calibrated two-stage designs. The comparison highlights two complementary messages. First, the proposed two-stage procedure can reproduce the desired Bayesian operating characteristics while incorporating an interim futility analysis. Second, the practical efficiency of the resulting design depends strongly on the design-prior specification used for calibration.

\begin{table}[t]
\centering
\caption{Summary of the riociguat example under the one-stage reference design and two two-stage calibrations.}
\label{tab:riociguat_summary}
\resizebox{\textwidth}{!}{%
\begin{tabular}{lcccccccc}
\hline
Design & \(n_{11}\) & \(n_{12}\) & \(n_{21}\) & \(n_{22}\) & \(N_{\max}\) & Power & Type-I error & \(E_{H_0}[N]\) \\
\hline
One-stage design & -- & -- & 27 & 26 & 53 & 0.80 & 0.0070 & 53.0 \\
Two-stage, mildly informative priors & 10 & 10 & 34 & 34 & 68 & 0.833 & 0.0058 & 66.04 \\
Two-stage, more informative priors & 10 & 10 & 13 & 13 & 26 & 0.853 & 0.0079 & 25.76 \\
\hline
\end{tabular}%
}
\end{table}

Overall, the riociguat example illustrates the central practical features of the proposed method. The algorithm provides a fully numerical, simulation-free calibration of Bayesian two-stage two-arm designs; it yields interpretable operating characteristics in terms of Bayes factor evidence thresholds; and it makes transparent how prior assumptions at the design stage influence both feasibility and efficiency. At the same time, the example also shows that early stopping for futility is not automatically associated with large savings in expected sample size. Whether such savings materialize depends crucially on how well the design priors separate the hypotheses under consideration.

In the riociguat example, the optimal two-stage design only stops early for
futility under $H_0$ with probability about \(0.04\), so the reduction in the
expected sample size under $H_0$ is very modest. This behaviour is not a bug
of the algorithm, but a consequence of the modelling choices and calibration
constraints.

First, the design is calibrated to fairly strict evidence requirements:
the success threshold \(k = 1/10\), the null-evidence threshold
\(k_f = 3\), the Bayesian type-I error bound \(\alpha = 0.025\),
and the requirement \(\mathrm{CE}_{H_0}^{\text{2st}} \ge 0.60\) together imply that
only a small fraction of $H_0$ outcomes can be eliminated safely at the
interim look without compromising either power or the probability of
compelling evidence in favour of $H_0$. Under such constraints, the interim
boundary cannot be very aggressive, so the early stopping probability under
$H_0$ remains low and \(E_{H_0}(N)\) stays close to the maximum sample size.

Second, even when the interim fraction is moved and the $\text{CE}_{H_0}^{\text{2st}}$ target is
varied, the futility probability in this example is relatively insensitive as
long as the thresholds \(k\) and \(k_f\) and the overall calibration targets
remain fixed. Moving the interim later increases the information available at
the interim, but the futility rule still has to preserve about 80\% Bayesian
power and the $\text{CE}_{H_0}^{\text{2st}}$ constraint, which limits how many null paths can be
stopped early. In particular, with \(k_f = 3\) already fairly liberal for
declaring evidence in favour of $H_0$, further gains in early
stopping would require relaxing this threshold in a way that is not clinically
desirable here.

Third, the design priors have a pronounced effect on the expected sample size
under $H_0$. When the design priors under \(H_+\) are made more informative
and more clearly separated from $H_0$, the predictive distributions under
$H_0$ and \(H_+\) diverge more quickly as the sample size grows. This leads
to a smaller sufficient fixed-sample size and, consequently, to a smaller
expected sample size under $H_0$ in the corresponding two-stage design, even
if the interim futility probability itself changes only marginally. In the
riociguat example, this can be achieved by concentrating the design priors
slightly more around the clinically relevant success rates, while keeping the
analysis priors and Bayes factor thresholds unchanged.

\subsection{Riociguat trial re-analysis with slightly informative design priors}
\label{subsec:riociguat-new}

We revisit the riociguat phase II trial, this time using less optimistic but still slightly informative design priors. This second example illustrates how the necessary sample sizes are influenced by the design prior choice and how the sample size reduction of the optimal two-stage design compared to the calibrated one-stage design varies with different design prior choices. In contrast to the first example, we not require 90\% Bayesian power instead of only 80\%. The type-I-error (Bayesian) is again calibrated to 2.5\%, and no minimum probability on compelling evidence in favour of $H_0$ is required.

\paragraph{Priors.}

We now use a slightly informative design prior under \(H_0\),
\[
p \mid H_0 \sim \mathrm{Beta}(a_{0d}, b_{0d}),
\]
and slightly informative design priors under \(H_1\),
\[
p_1 \mid H_1 \sim \mathrm{Beta}(1,3),
\qquad
p_2 \mid H_1 \sim \mathrm{Beta}(3,1),
\]
encoding the expectation that the control arm has a lower response probability than the treatment arm. The corresponding analysis priors are chosen to be flat,
\[
p \mid H_0 \sim \mathrm{Beta}(1,1),
\qquad
p_1 \mid H_1 \sim \mathrm{Beta}(1,1),
\qquad
p_2 \mid H_1 \sim \mathrm{Beta}(1,1).
\]
This choice reflects the intended separation between planning and analysis: prior information is allowed to influence the calibration of the design, but the eventual evidential assessment through the Bayes factor is based on neutral analysis priors.

\paragraph{One-stage reference design.}

As a benchmark, we first consider the fixed-sample one-stage design obtained under the same thresholds and Bayesian calibration targets. Under the priors specified above and balanced allocation between arms, the one-stage calibration identifies a design with a total sample size of
\[
N_{\text{one}} = 154
\]
patients, corresponding to 77 patients in each arm. At this sample size the Bayesian power is approximately \(0.901\), the Bayesian type-I error under \(H_0\) is about \(0.004\), and the probability of compelling evidence for \(H_0\) is about \(0.775\). This directly calibrated one-stage design serves as a reference for evaluating the corresponding two-stage design. The results are shown in \Cref{fig:one-stage-modified}.

\begin{figure}[h!]
    \centering
    \includegraphics[width=1\linewidth]{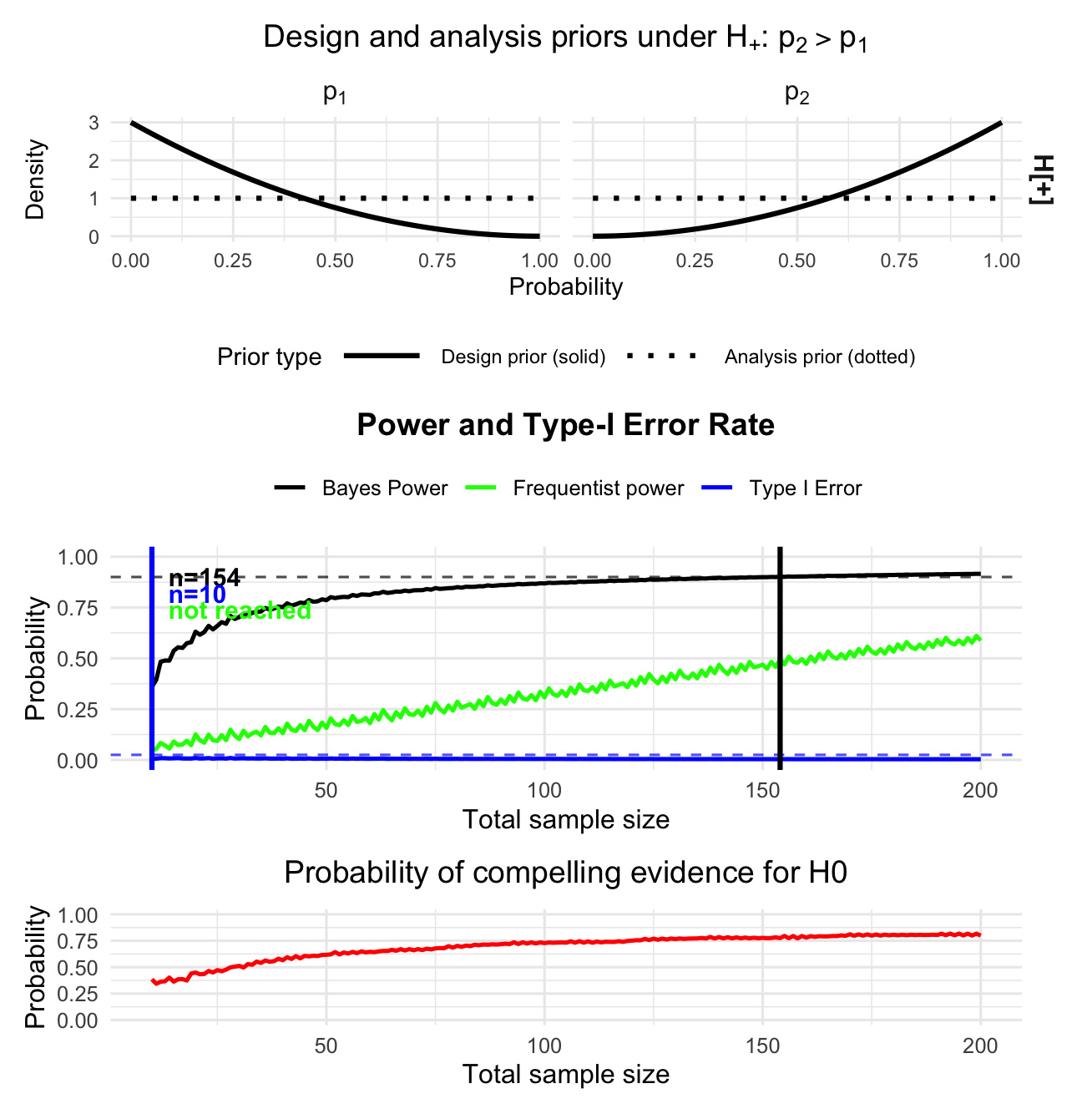}
    \caption{Calibrated one-stage Bayes factor design for the riociguat example. The figure illustrates the operating characteristics of the fixed-sample reference design under the chosen evidence thresholds and prior specification. Slightly informative design priors $B(1,3)$ and $B(3,1)$ were chosen for the control and treatment group.}
    \label{fig:one-stage-modified}
\end{figure}

\paragraph{Two-stage design without power cushion.}

We now apply the proposed two-stage calibration algorithm. The design includes a single interim analysis that allows early stopping for futility and is calibrated to the same Bayesian power and type-I-error targets as the one-stage design. Also, no calibration of the compelling evidence in favour of $H_0$ is carried out. The efficacy and futility thresholds \(k\) and \(k_f\) and the priors are kept unchanged. 

Under these settings, the fixed-sample calibration step identifies a sufficient one-stage anchor with
\[
n_{2}^{(1)} = n_{2}^{(2)} = 76,
\]
corresponding to a total sample size of \(N_{\text{anchor}} = 152\). At this anchor the Bayesian power under the design priors is approximately \(0.900\), the Bayesian type-I error is about \(0.004\), and the Bayesian probability of compelling evidence for \(H_0\) is about \(0.776\). Conditional on this anchor, the second step of the algorithm searches over admissible interim sample sizes and selects the design that minimizes the expected total sample size under \(H_0\). The results are shown in \Cref{fig:two-stage-modified}.

\begin{figure}[h!]
    \centering
    \includegraphics[width=1\textwidth]{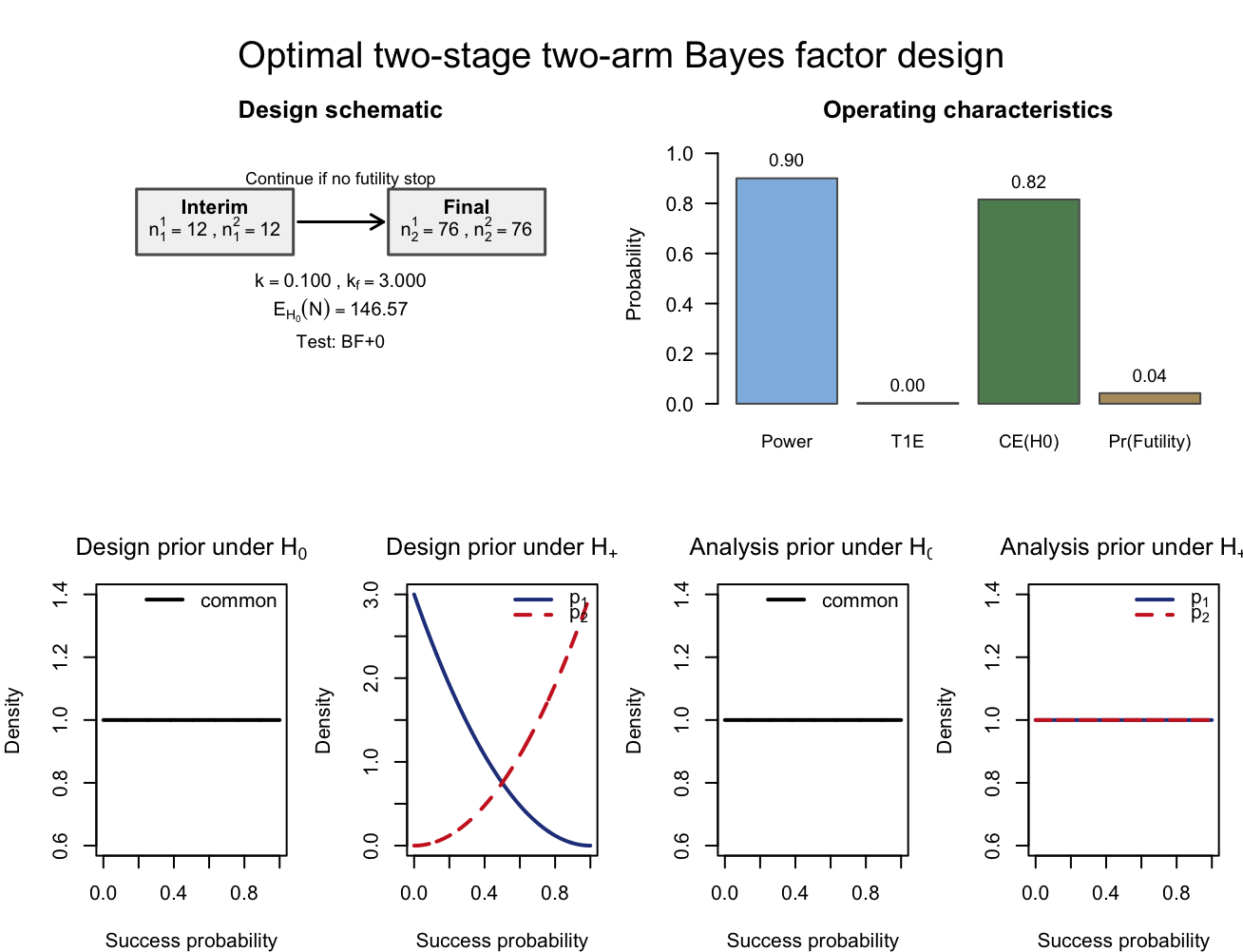}
    \caption{Optimal two-stage Bayes factor design for the riociguat example using more informative design priors under \(H_1\) and flat analysis priors. The figure visualizes the calibrated design and the corresponding prior specification used in planning.}
    \label{fig:two-stage-modified}
\end{figure}
For the riociguat example, the resulting optimal two-stage design is
\[
(n_{1}^{(1)}, n_{1}^{(2)}, n_{2}^{(1)}, n_{2}^{(2)}) = (12, 12, 76, 76).
\]
Hence, the interim analysis is conducted after \(24\) patients in total, and the maximal sample size is \(N_{\max} = 152\). The corrected Bayesian operating characteristics of this design, accounting for early stopping for futility, are
\[
\text{Power} \approx 0.900,\qquad
\text{Type-I error} \approx 0.003,\qquad
\mathrm{CE}_{H_0} \approx 0.816.
\]
The corresponding Bayesian expected sample size under \(H_0\) is
\[
E_{H_0}[N] \approx 146.6.
\]

\paragraph{Comparison and interpretation.}

Compared to the one-stage reference design with \(N_{\text{one}} = 154\) and no interim analysis, the two-stage design slightly reduces the maximal sample size to \(N_{\max} = 152\) and, more importantly, reduces the expected sample size under \(H_0\) from \(154\) to about \(146.6\), while maintaining essentially the same Bayesian power and type-I-error rate. The probability of compelling evidence for \(H_0\) is actually increased from approximately \(0.775\) to \(0.816\), reflecting the fact that interim futility stopping is counted as compelling evidence in favour of the null.

A subtle point in this example is that the fixed-sample anchor identified in step~1 of the two-stage calibration (76 patients per arm) does not exactly coincide with the smallest feasible one-stage design (77 patients per arm). This discrepancy is not driven by the \(\mathrm{CE}_{H_0}\) constraint, which is inactive here, but rather by two technical aspects of the calibration: the Bayesian power and type-I-error functions under the beta--binomial design priors are not strictly monotone in the total sample size on the integer grid, and the one-stage calibration algorithm enforces sustained feasibility over a grid of parameter values, whereas the two-stage engine only requires a single fixed-sample size to meet the marginal Bayesian targets in step~1. This implies that the one-stage calibration assures that e.g. the Bayesian power does not drop below the specified target constraint for at least the next 10 observations (analogue for type-I-error rate and probability of compelling evidence for $H_0$).\footnote{Note that such a sustained calibration logic makes no sense in the two-stage design, as interplay between the position of the interim analysis and the oscillations in the beta-binomial model lead to a situation where there is no monotone relationship for the relevant operating characteristics anymore. For example, power is not necessarily a (strictly) increasing function of the interim position for a fixed final sample size.} Together with small oscillations in the grid-based beta--binomial calculations, this leads to a situation where the two-stage algorithm accepts \(n_{2}^{(1)} = n_{2}^{(2)} = 76\) as a sufficient anchor, while the one-stage search reports \(n_1 = n_2 = 77\) as the smallest sustained-feasible design. Conditional on this anchor, the two-stage design then preserves the desired Bayesian power and type-I error and achieves a smaller expected sample size under \(H_0\).

Overall, the riociguat example shows that, with realistic directional Bayes factors and slightly informative design priors, the proposed two-stage calibration algorithm can introduce an interim futility analysis and achieve a genuine reduction in the expected sample size under the null hypothesis, without compromising the Bayesian power and type-I-error targets.

\section{Discussion}
\label{sec:discussion}

This paper developed a simulation-free methodology for Bayesian optimal two-stage designs in two-arm phase II clinical trials with binary endpoints using Bayes factors. The proposed approach combines the matrix-search framework available for fixed-sample two-arm Bayes factor designs with an exact correction for interim futility stopping, thereby extending simulation-free Bayesian calibration from one-arm and fixed-sample settings to the practically important case of two-arm two-stage designs.

\subsection{Benefits}
A central strength of the method is that all relevant operating characteristics can be computed by finite summation over prior-predictive probabilities rather than by Monte Carlo simulation. This has several advantages.

\begin{itemize}
    \item[$\blacktriangleright$]{First, calibration is reproducible and deterministic, because the operating characteristics do not depend on simulation size, random seeds, or Monte Carlo error. The calibration algorithm takes only seconds to find an optimal design on a regular personal computer.}
    \item[$\blacktriangleright$]{Second, the resulting design is transparent: for a given prior specification, Bayes factor threshold, and admissible sample size region, the final design can be traced back directly to the corresponding efficacy and futility regions in the discrete sample space.}
    \item[$\blacktriangleright$]{Third, the methodology is practically useful because the same framework supports design selection, sensitivity analyses, and interpretation of how individual modelling choices influence the resulting operating characteristics.}
    \item[$\blacktriangleright$]{Fourth, the methodology is simple to apply. Due to the nearly instantaneous computing time, visualizing the results and the prior assumptions with the \texttt{bfbin2arm} package as shown in the plots in \Cref{sec:examples} provides a convenient way to design a phase II trial in practice.}
    \item[$\blacktriangleright$]{Fifth, the methodology allows to calibrate the probability of compelling evidence for $H_0$. Thus, when $H_0$ holds true, a design can be planned and calibrated so that there is at least probability $f$ to find compelling evidence for $H_0$ and not end up with an indecisive result.}
    \item[$\blacktriangleright$]{Sixth, the methodology is like a modular system. Target constraints on the Bayesian power, type-I-error rate and probability of compelling evidence work independently of each other and can be added or removed from a design in the planning stage when running the calibration. This allows for flexible design planning and comparisons of different (increasingly restricting) requirements on the relevant operating characteristics.}
\end{itemize}  
At the same time, the proposed design is not purely ``objective'' in the planning stage, because the operating characteristics depend on the design priors. This is not a weakness of the Bayesian framework but an inherent feature of prior-predictive calibration. The design priors determine how plausibly separated the hypotheses are before data are observed, and therefore how quickly the Bayes factor is expected to accumulate evidence as the sample size increases. In contrast, the analysis priors enter the Bayes factor used at interim and final analysis. Separating design and analysis priors is therefore conceptually important. It allows substantive prior knowledge or planning assumptions to influence the design, while the eventual analysis can still be based on comparatively neutral priors. The riociguat example shows that this distinction is not merely philosophical: stronger prior-predictive separation at the design stage can markedly reduce the sample size required for calibration, even when the analysis priors and Bayes factor thresholds are kept fixed. A word of caution should be added for users who think flat design priors are desirable: First, in most cases flat design priors are unrealistic from a scientific perspective, because extremely large and extremely small success probabilities are equally likely a priori as moderately sizes success probabilities. Second, flat design priors lead to a larger sufficient sample size found in the first step of the calibration algorithm. This in turn increases the number of two-stage designs which need to be analyzed in step two of the calibration algorithm, substantially increasing runtime. Thus, both from a scientific and computational perspective, flat design priors are undesirable. We strongly recommend using slightly informative design priors which accurately reflect the expectation about the treatment effect and are neither too optimistic nor too pessimistic in that sense.

\subsection{Efficiency gains and runtime}
The examples also show that introducing an interim analysis does not automatically imply a large reduction in expected sample size under \(H_0\). In the original riociguat calibration, the probability of early stopping for futility under \(H_0\) is only about \(4\%\), and the expected sample size remains close to the maximal sample size. This behaviour is a direct consequence of the evidential constraints imposed on the design. If efficacy requires strong evidence against \(H_0\), futility requires non-trivial evidence in favour of \(H_0\), and the probability of compelling evidence for \(H_0\) must also exceed a prespecified lower bound, then only a restricted subset of null trajectories can be truncated early without violating the design targets. In that sense, the method makes an important trade-off explicit: tighter evidence requirements can improve interpretability, but they generally reduce the room for aggressive early stopping.

A related practical lesson concerns the effect of the Bayes factor thresholds and the calibration targets on feasibility and runtime. Stringent efficacy thresholds such as \(k=1/10\) typically require larger fixed-sample anchor designs in Step~1 of the algorithm, because stronger evidence is needed before the trial can declare success. Likewise, demanding large values of \(\mathrm{CE}_{H_0}\) can force the calibration procedure to continue to larger sample sizes, since small designs may simply not contain enough information for the Bayes factor to provide compelling support for the null hypothesis. These larger fixed-sample anchors then expand the set of admissible interim designs that must be evaluated in Step~2, thereby increasing runtime. Thus, the statistical and computational aspects of calibration are closely linked: stricter evidence demands may be scientifically desirable, but they also enlarge the search problem.

The choice of design priors has a similarly important computational consequence. Diffuse priors under \(H_0\) and \(H_1\) tend to make the prior-predictive distributions overlap more strongly, so that larger sample sizes are needed for the Bayes factor to distinguish the competing hypotheses reliably. In the current algorithm, this means that very flat design priors can substantially enlarge the fixed-sample anchor found in Step~1 and, through this, enlarge the interim design grid explored in Step~2. By contrast, moderately informative design priors that reflect clinically plausible response rates can both improve interpretability and reduce runtime. For practical applications, this suggests that sensitivity analyses over a range of scientifically credible design priors should be treated as part of routine design work rather than as an optional afterthought.

Another important contribution of the present work is interpretability. The corrected operating characteristics have a direct trial-level meaning. The corrected Bayesian power quantifies the probability, under the design prior for \(H_1\), that the two-stage procedure ultimately concludes in favour of treatment efficacy. The corrected Bayesian type-I error quantifies the corresponding probability under \(H_0\). The probability of compelling evidence for \(H_0\) and the expected sample size under \(H_0\) complement these quantities by describing how the design behaves when the treatment is ineffective. In a phase II setting, where the main goals are screening, learning, and avoiding unnecessary continuation of ineffective treatments, these quantities are often more informative than a single classical error-rate statement.

\subsection{Limitations}
The paper also has limitations. First, the methodology currently focuses on a single interim analysis with early stopping for futility only. This is already a useful and clinically relevant class of designs, but it does not cover multiple interim looks or early stopping for efficacy. Second, the current framework is developed for binary endpoints under beta--binomial modelling. Many phase II trials involve time-to-event, continuous, ordinal, or composite outcomes, and extending simulation-free Bayes factor calibration to such settings will require additional methodological work. Third, although the procedure is simulation-free, the computational burden can still become substantial when the admissible sample size region is large or when thresholds and priors lead to large fixed-sample anchors. The method therefore replaces Monte Carlo uncertainty by deterministic but potentially non-trivial numerical search.

\subsection{Future research}
Several directions for future research follow naturally from these limitations. One important extension would be to allow early stopping for efficacy in addition to futility, which would require corresponding corrections for efficacy-erased trajectories and a revised calibration criterion. A second extension would be to consider designs with more than one interim analysis, where the path structure becomes richer and the correction terms correspondingly more complex. A third direction would be to investigate additional optimization criteria beyond minimizing \(E_{H_0}[N]\), for example weighted average sample size criteria, minimax-type criteria, or utility-based criteria that reflect different clinical priorities. Finally, extensions to other endpoint types and more complex randomization schemes would broaden the applicability of the approach in real trial settings.

The most important relevant extension of the current work possibly is to introduce different calibration modes into the optimal design routine. This paper focussed on Bayesian operating characteristics, but regulatory agencies often require strict frequentist type-I-error control under $H_0$. Thus, a frequentist two-stage optimal design based on Bayes factors would be a possible extension. Here, the Bayes factor is used primarily as a test statistic whose frequentist -- not Bayesian -- operating characteristics such as frequentist type-I-error rate and power are relevant for the calibration. Likewise, hybrid or even full calibration modes where both frequentist and Bayesian operating characteristics must simultaneously be calibrated could extend the current work.

\subsection{Summary}
In summary, the proposed methodology provides a principled and computationally reproducible framework for designing Bayesian two-arm two-stage phase II trials with binary endpoints using Bayes factors. Its main practical message is twofold. On the one hand, simulation-free calibration is feasible even in this more complex two-stage two-arm setting. On the other hand, the efficiency of the resulting design depends crucially on scientifically meaningful prior specification and on the interplay between evidence thresholds, calibration targets, and admissible sample size regions. For phase II applications in which Bayes factors are viewed as the primary evidential measure, this framework offers a transparent basis for design calibration and sensitivity analysis.

\section*{Acknowledgements}
The author is grateful to Silke Jörgens, Kathrin Möllenhoff and Samuel Pawel for helpful comments, discussions and suggestions on the methodology developed in this manuscript.

\appendix
\appendixpage

\section{The Appendix}
\subsection{Proofs}
\begin{proof}[Proof of Lemma 1 (Conditional Independence)]
Under the design prior $\pi^{(i)}(p_1,p_2)$, the counts $X_1,X_2,Z_1,Z_2$ are conditionally independent given $(p_1,p_2)$, with
\[
  X_j \sim \text{Bin}(n_1^{(j)}, p_j), \quad
  Z_j \sim \text{Bin}(n_2^{(j)} - n_1^{(j)}, p_j), \quad j = 1,2.
\]
Therefore, the conditional probability mass function factorizes as
\begin{align*}
  \Pr(X_1 = x_1, X_2 = x_2, Z_1 = z_1, Z_2 = z_2 \mid p_1,p_2)
  &= \text{Bin}(x_1 \mid n_1^{(1)}, p_1)\cdot 
    \text{Bin}(x_2 \mid n_1^{(2)}, p_2)\\
    &\cdot
    \text{Bin}(z_1 \mid n_2^{(1)} - n_1^{(1)}, p_1)\cdot
    \text{Bin}(z_2 \mid n_2^{(2)} - n_1^{(2)}, p_2).
\end{align*}
The joint prior-predictive probability mass function is obtained by marginalizing over the design prior:
\begin{align*}
  \pi^{(i)}(x_1,x_2,z_1,z_2)
  &= \int_{p_1,p_2}
    \text{Bin}(x_1 \mid n_1^{(1)}, p_1)\cdot
    \text{Bin}(x_2 \mid n_1^{(2)}, p_2)\\
    &\cdot\text{Bin}(z_1 \mid n_2^{(1)} - n_1^{(1)}, p_1)
    \cdot\text{Bin}(z_2 \mid n_2^{(2)} - n_1^{(2)}, p_2)
    \cdot\pi^{(i)}(p_1,p_2)
  \text{d}p_1 \text{d}p_2.
\end{align*}
Because the integrand is a product of a function of $(x_1,x_2)$ and a function of $(z_1,z_2)$, and the prior $\pi^{(i)}(p_1,p_2)$ is shared, the integral factorizes as
\begin{align*}
  \pi^{(i)}(x_1,x_2,z_1,z_2)
  &= \underbrace{
      \int_{p_1,p_2}
        \text{Bin}(x_1 \mid n_1^{(1)}, p_1)
        \text{Bin}(x_2 \mid n_1^{(2)}, p_2)
        \pi^{(i)}(p_1,p_2)
      \text{d}p_1 \text{d}p_2
    }_{:= f_1(x_1,x_2 \mid H_i)}\\
    &\cdot
    \underbrace{
      \int_{p_1,p_2}
        \text{Bin}(z_1 \mid n_2^{(1)} - n_1^{(1)}, p_1)
        \text{Bin}(z_2 \mid n_2^{(2)} - n_1^{(2)}, p_2)
        \pi^{(i)}(p_1,p_2)
      \text{d}p_1 \text{d}p_2
    }_{:= f_2(z_1,z_2 \mid H_i)}.
\end{align*}
This proves the factorization
\[
  \pi^{(i)}(x_1,x_2,z_1,z_2)
  = f_1(x_1,x_2 \mid H_i) \cdot f_2(z_1,z_2 \mid H_i).
\]
\end{proof}

\begin{proof}[Proof of Theorem 1 (Double-sum expression of the futility-erased partial contribution)]
We now show the double‑sum expression for the futility‑erased partial contribution. Let
\[
  \Delta^{(i)}
  := \sum_{(x_1,x_2)\in \mathcal{F}_1^{(i)}}
     \sum_{\substack{(z_1,z_2):\\ (x_1+z_1, x_2+z_2)\in \mathcal{E}_2^{(i)}}}
     \pi^{(i)}(x_1,x_2,z_1,z_2),
\]
where $\mathcal{F}_1^{(i)}$ is the interim futility region and $\mathcal{E}_2^{(i)}$ is the final‑stage efficacy region. Using the above factorization of \Cref{lemma:1},
\[
  \pi^{(i)}(x_1,x_2,z_1,z_2)
  = f_1(x_1,x_2 \mid H_i) \cdot f_2(z_1,z_2 \mid H_i),
\]
so
\[
  \Delta^{(i)}
  = \sum_{(x_1,x_2)\in \mathcal{F}_1^{(i)}}
      \sum_{\substack{(z_1,z_2):\\ (x_1+z_1, x_2+z_2)\in \mathcal{E}_2^{(i)}}}
      f_1(x_1,x_2 \mid H_i)
      f_2(z_1,z_2 \mid H_i).
\]
Since $f_1(x_1,x_2 \mid H_i)$ does not depend on $(z_1,z_2)$, it can be factored outside the inner sum:
\[
  \Delta^{(i)}
  = \sum_{(x_1,x_2)\in \mathcal{F}_1^{(i)}}
      f_1(x_1,x_2 \mid H_i)
      \Biggl(
        \sum_{\substack{(z_1,z_2):\\ (x_1+z_1, x_2+z_2)\in \mathcal{E}_2^{(i)}}}
          f_2(z_1,z_2 \mid H_i)
      \Biggr).
\]
which is the double‑sum expression for the futility‑erased partial contribution to the Bayesian power (or type‑I‑error rate) given in \Cref{theorem:1}. The inner sum corresponds to the conditional probability that, if the trial were continued from interim counts $(x_1,x_2)\in \mathcal{F}_1^{(i)}$, the final data would fall into the efficacy region $\mathcal{E}_2^{(i)}$, under the design prior and given the interim information.
\end{proof}

\subsection{\Cref{lemma:1} and \Cref{theorem:1} for directional hypothesis tests}
\begin{lemma}[Factorization and futility-erased contribution under directional tests]\label{lemma:2}
Under the directional tests
\begin{align}
    &H_0: \eta \leq 0 \quad \text{versus} \quad H_1: \eta > 0,\label{eq:directionalParameterized}\\
    &H_0: \eta = 0 \quad \text{versus} \quad H_1: \eta > 0,\label{eq:directionalParameterizedOneSided1}\\
    &H_0: \eta = 0 \quad \text{versus} \quad H_1: \eta < 0,\label{eq:directionalParameterizedOneSided2}
\end{align}
with truncated beta design priors under $H_1$ such that $\eta$ is constrained as above, the conditional independence of the first- and second-stage counts $X_1,X_2,Z_1,Z_2$ given $(p_1,p_2)$ remains unchanged. As a consequence, the joint prior-predictive pmf factorizes as
\[
  \pi^{(i,\text{dir})}(x_1,x_2,z_1,z_2)
  = f_1^{\text{dir}}(x_1,x_2 \mid H_i)
    \cdot
    f_2^{\text{dir}}(z_1,z_2 \mid H_i),
\]
for $i = 0,1$, where $f_1^{\text{dir}}$ and $f_2^{\text{dir}}$ are the prior-predictive probability mass functions computed under the truncated directional design priors, compare \cite{kelterTwoArmTwoStage2026}.
\end{lemma}
\begin{proof}
Under the sampling model $X_j \sim \text{Bin}(n_1^{(j)},p_j)$ and $Z_j \sim \text{Bin}(n_2^{(j)} - n_1^{(j)},p_j)$, $j=1,2$, the counts $X_1,X_2,Z_1,Z_2$ are conditionally independent given $(p_1,p_2)$, and their joint conditional probability mass function factors as
\begin{align*}
  \Pr(X_1 = x_1, X_2 = x_2,& Z_1 = z_1, Z_2 = z_2 \mid p_1,p_2)
  = \text{Bin}(x_1 \mid n_1^{(1)},p_1)\cdot
    \text{Bin}(x_2 \mid n_1^{(2)},p_2)\\
    &
    \cdot \text{Bin}(z_1 \mid n_2^{(1)} - n_1^{(1)},p_1)
    \cdot \text{Bin}(z_2 \mid n_2^{(2)} - n_1^{(2)},p_2).
\end{align*}
Marginalizing over the truncated directional design prior $\pi^{(i,\text{dir})}(p_1,p_2)$ yields the joint prior-predictive probability mass function
\begin{align*}
  \pi^{(i,\text{dir})}(x_1,x_2,z_1,z_2)
  = \int_{p_1,p_2}
      &\text{Bin}(x_1 \mid n_1^{(1)},p_1)\cdot
      \text{Bin}(x_2 \mid n_1^{(2)},p_2)\cdot
      \text{Bin}(z_1 \mid n_2^{(1)} - n_1^{(1)},p_1)\\\cdot
      &\text{Bin}(z_2 \mid n_2^{(2)} - n_1^{(2)},p_2)\cdot
      \pi^{(i,\text{dir})}(p_1,p_2)
    \text{d}p_1 \text{d}p_2.
\end{align*}
Because the integrand is a product of a function of $(x_1,x_2)$ and a function of $(z_1,z_2)$, and both share the same truncated prior over $(p_1,p_2)$, the integral factorizes as
\begin{align*}
  \pi^{(i,\text{dir})}(x_1,x_2,z_1,z_2)
  &= \underbrace{
       \int_{p_1,p_2}
         \text{Bin}(x_1 \mid n_1^{(1)},p_1)
         \text{Bin}(x_2 \mid n_1^{(2)},p_2)
         \pi^{(i,\text{dir})}(p_1,p_2)
       \text{d}p_1 \text{d}p_2
     }_{=: f_1^{\text{dir}}(x_1,x_2 \mid H_i)}
     \\
     &\quad \cdot
     \underbrace{
       \int_{p_1,p_2}
         \text{Bin}(z_1 \mid n_2^{(1)} - n_1^{(1)},p_1)
         \text{Bin}(z_2 \mid n_2^{(2)} - n_1^{(2)},p_2)
         \pi^{(i,\text{dir})}(p_1,p_2)
       \text{d}p_1 \text{d}p_2
     }_{=: f_2^{\text{dir}}(z_1,z_2 \mid H_i)}.
\end{align*}
which proves the factorization claimed in \Cref{lemma:2}.
\end{proof}

\begin{theorem}\label{theorem:2}
Under the conditions of \Cref{lemma:2}, the futility-erased partial contribution to the Bayesian power (or type-I-error rate) then admits the double-sum representation
\[
  \Delta^{(i)}
  = \sum_{(x_1,x_2)\in\mathcal{F}_1^{(i)}}f_1^{\text{dir}}(x_1,x_2 \mid H_i)
     \sum_{\substack{(z_1,z_2):\\ (x_1+z_1,x_2+z_2)\in\mathcal{E}_2^{(i)}}}
     f_2^{\text{dir}}(z_1,z_2 \mid H_i),
\]
where $\mathcal{F}_1^{(i)}$ is the interim futility region and $\mathcal{E}_2^{(i)}$ is the final-stage efficacy region defined in terms of the corresponding directional Bayes factors.
\end{theorem}
\begin{proof}
By definition, the futility-erased partial contribution is the total probability that
\begin{itemize}
  \item the interim data $(x_1,x_2)$ fall into the futility region $\mathcal{F}_1^{(i)}$, and
  \item had the trial continued, the final data $(y_1,y_2) = (x_1+z_1,x_2+z_2)$ would have fallen into the final-stage efficacy region $\mathcal{E}_2^{(i)}$.
\end{itemize}
Formally, this is
\[
  \Delta^{(i)}
  = \sum_{(x_1,x_2)\in\mathcal{F}_1^{(i)}}
     \sum_{\substack{(z_1,z_2):\\ (x_1+z_1,x_2+z_2)\in\mathcal{E}_2^{(i)}}}
     \pi^{(i,\text{dir})}(x_1,x_2,z_1,z_2).
\]
Using the factorization from \Cref{lemma:2}, 
\[
  \pi^{(i,\text{dir})}(x_1,x_2,z_1,z_2)
  = f_1^{\text{dir}}(x_1,x_2 \mid H_i)
    \cdot
    f_2^{\text{dir}}(z_1,z_2 \mid H_i),
\]
so
\[
  \Delta^{(i)}
  = \sum_{(x_1,x_2)\in\mathcal{F}_1^{(i)}}f_1^{\text{dir}}(x_1,x_2 \mid H_i)
      \sum_{\substack{(z_1,z_2):\\ (x_1+z_1,x_2+z_2)\in\mathcal{E}_2^{(i)}}}
      f_2^{\text{dir}}(z_1,z_2 \mid H_i).
\]
which is the double-sum representation claimed in \Cref{theorem:2}.
\end{proof}

\section{Correction of the probability of compelling evidence under futility-only interim monitoring in two-stage designs}\label{sec:corrections}

Consider a two-arm trial with binary endpoints. Let \(n_1^1\) and \(n_1^2\) denote the interim sample sizes in the control and treatment arms, respectively, and let \(n_2^1\) and \(n_2^2\) denote the corresponding final sample sizes, with \(n_1^j \le n_2^j\) for \(j \in \{1,2\}\). Let \(X_1\) and \(X_2\) denote the interim numbers of responses in the two arms, and let \(Z_1\) and \(Z_2\) denote the additional responses accrued between interim and final analysis. Hence the final response counts are
\[
Y_1 = X_1 + Z_1, \qquad Y_2 = X_2 + Z_2.
\]

Assume that a futility-only interim analysis is performed. If the interim Bayes factor provides compelling evidence for \(H_0\), recruitment is stopped early; otherwise the trial continues to the planned final sample size. This stopping rule changes the operating characteristics of the design relative to the corresponding fixed-sample design.

In particular, the fixed-sample power and type-I error are no longer valid for the two-stage design, because some trajectories that would have yielded final efficacy under the fixed-sample design are removed by early stopping for futility. Therefore, the fixed-sample power and type-I error must be corrected by subtracting the probability mass of those erased trajectories.

By contrast, the probability of compelling evidence for \(H_0\) increases under futility-only interim monitoring. The reason is that, under the two-stage design, any interim outcome that already yields compelling evidence for \(H_0\) is counted immediately as success for \(H_0\). Under the corresponding fixed-sample design, those same interim outcomes would not stop the trial, and some of them would fail to yield compelling evidence for \(H_0\) at the final analysis. Consequently, the two-stage probability of compelling evidence for \(H_0\) equals the fixed-sample probability plus the probability mass of interim-futility trajectories that would not have ended in compelling evidence for \(H_0\) at the final analysis.

\subsection{Explicit two-arm correction formula for \texorpdfstring{\( \mathrm{CE}_{H_0} \)}{CE(H0)}}\label{sec:correctionFutilityStopping}

Let \(F_1\) denote the interim futility region, i.e.
\[
F_1 = \left\{ (x_1,x_2) : \mathrm{BF}_{\mathrm{int}}(x_1,x_2) \ge k_f \right\},
\]
where \(k_f\) is the futility threshold and \(\mathrm{BF}_{\mathrm{int}}\) denotes the Bayes factor evaluated at the interim sample size. Let \(C_2\) denote the final compelling-evidence region for \(H_0\), i.e.
\[
C_2 = \left\{ (y_1,y_2) : \mathrm{BF}_{\mathrm{fin}}(y_1,y_2) \ge k_f \right\},
\]
where \(\mathrm{BF}_{\mathrm{fin}}\) denotes the Bayes factor evaluated at the final sample size. Let
\[
f_{1,0}(x_1,x_2)
\]
be the joint prior-predictive probability mass function of the interim responses under \(H_0\), and let
\[
f_{2,0}(z_1,z_2)
\]
be the joint prior-predictive probability mass function of the stage-2 increments under \(H_0\). Under conditional independence of stage-1 and stage-2 increments given the design prior under \(H_0\), the fixed-sample probability of compelling evidence for \(H_0\) is
\[
\mathrm{CE}^{\mathrm{fix}}_{H_0}
=
\sum_{y_1=0}^{n_2^1}\sum_{y_2=0}^{n_2^2}
\mathbf{1}\!\left\{(y_1,y_2)\in C_2\right\}
\, f^{\mathrm{fix}}_{0}(y_1,y_2),
\]
where \(f^{\mathrm{fix}}_{0}\) is the joint prior-predictive distribution at the final sample size. For the corresponding two-stage design with futility-only stopping, the corrected probability of compelling evidence for \(H_0\) is
\begin{align}\label{eq:pce2st}
\mathrm{CE}^{\mathrm{2st}}_{H_0}
=
\mathrm{CE}^{\mathrm{fix}}_{H_0}
+
\Delta_{\mathrm{CE},0},
\end{align}
where
\begin{align}
\Delta_{\mathrm{CE},0}
&=
\sum_{x_1=0}^{n_1^1}\sum_{x_2=0}^{n_1^2}
\mathbf{1}\!\left\{(x_1,x_2)\in F_1\right\}
f_{1,0}(x_1,x_2)\nonumber\\
&\cdot\left[
1
-
\sum_{z_1=0}^{n_2^1-n_1^1}\sum_{z_2=0}^{n_2^2-n_1^2}
\mathbf{1}\!\left\{(x_1+z_1,x_2+z_2)\in C_2\right\}
f_{2,0}(z_1,z_2)
\right].\nonumber
\end{align}
This term collects exactly those interim-futility trajectories that are counted as compelling evidence for \(H_0\) in the two-stage design, but that would \emph{not} have yielded compelling evidence for \(H_0\) at the final analysis under the corresponding fixed-sample design.

Equivalently, one may write
\[
\mathrm{CE}^{\mathrm{2st}}_{H_0}
=
\sum_{x_1=0}^{n_1^1}\sum_{x_2=0}^{n_1^2}
\mathbf{1}\!\left\{(x_1,x_2)\in F_1\right\} f_{1,0}(x_1,x_2)
+
\sum_{x_1=0}^{n_1^1}\sum_{x_2=0}^{n_1^2}
\mathbf{1}\!\left\{(x_1,x_2)\notin F_1\right\} f_{1,0}(x_1,x_2)
\]
\[
\qquad\qquad \times
\sum_{z_1=0}^{n_2^1-n_1^1}\sum_{z_2=0}^{n_2^2-n_1^2}
\mathbf{1}\!\left\{(x_1+z_1,x_2+z_2)\in C_2\right\}
f_{2,0}(z_1,z_2).
\]
The first term is the probability of stopping early for futility, while the second term is the probability of reaching compelling evidence for \(H_0\) at the final analysis after continuation.

Hence,
\[
\mathrm{CE}^{\mathrm{2st}}_{H_0} \ge \mathrm{CE}^{\mathrm{fix}}_{H_0},
\]
with strict inequality whenever there exists positive prior-predictive probability for an interim outcome in \(F_1\) that would not end in \(C_2\) under continuation to the final sample size.

\subsection{Implication for the optimal design calibration algorithm}

Therefore, in a two-stage futility-only Bayes-factor design, the probability of compelling evidence for \(H_0\) must be calibrated using the corrected two-stage quantity \( \mathrm{CE}^{\mathrm{2st}}_{H_0} \), rather than the fixed-sample quantity \( \mathrm{CE}^{\mathrm{fix}}_{H_0} \). In particular, if a design constraint
\[
\mathrm{CE}^{\mathrm{2st}}_{H_0} \ge p_{\mathrm{CE},H_0}
\]
is imposed, then this constraint must be checked during the second-stage calibration over admissible interim designs, together with the corrected power and corrected type-I error constraints.

\bibliography{library}

@article{Bartos2022,
  title = {Informed {{Bayesian}} Survival Analysis},
  author = {Barto{\v s}, Franti{\v s}ek and Aust, Frederik and Haaf, Julia M.},
  year = 2022,
  month = sep,
  journal = {BMC Medical Research Methodology 2022 22:1},
  volume = {22},
  number = {1},
  pages = {1--22},
  publisher = {BioMed Central},
  issn = {1471-2288},
  doi = {10.1186/S12874-022-01676-9},
  urldate = {2023-02-15},
  pmid = {36088281}
}

@article{Berry2006,
  title = {Bayesian Clinical Trials},
  author = {Berry, Donald A.},
  year = 2006,
  month = feb,
  journal = {Nature Reviews Drug Discovery 2006 5:1},
  volume = {5},
  number = {1},
  pages = {27--36},
  publisher = {Nature Publishing Group},
  issn = {1474-1784},
  doi = {10.1038/nrd1927},
  urldate = {2022-02-02},
  pmid = {16485344}
}

@book{Berry2011,
  title = {Bayesian {{Adaptive Methods}} for {{Clinical Trials}}},
  author = {Berry, Scott M.},
  year = 2011,
  publisher = {CRC Press},
  address = {Boca Raton, FL},
  urldate = {2022-02-01},
  isbn = {978-0-429-15242-9}
}

@article{Boulesteix2020,
  title = {Introduction to Statistical Simulations in Health Research},
  author = {Boulesteix, Anne Laure and Groenwold, Rolf H.H. and Abrahamowicz, Michal and Binder, Harald and Briel, Matthias and Hornung, Roman and Morris, Tim P. and Rahnenf{\"u}hrer, J{\"o}rg and Sauerbrei, Willi},
  year = 2020,
  month = dec,
  journal = {BMJ Open},
  volume = {10},
  number = {12},
  pages = {e039921},
  publisher = {British Medical Journal Publishing Group},
  issn = {2044-6055},
  doi = {10.1136/BMJOPEN-2020-039921},
  urldate = {2022-03-08},
  pmid = {33318113}
}

@article{Chevret2012,
  title = {Bayesian Adaptive Clinical Trials: {{A}} Dream for Statisticians Only?},
  author = {Chevret, Sylvie},
  year = 2012,
  journal = {Statistics in Medicine},
  volume = {31},
  number = {11-12},
  pages = {1002--1013},
  publisher = {John Wiley \& Sons, Ltd},
  issn = {02776715},
  doi = {10.1002/sim.4363},
  urldate = {2022-02-01},
  isbn = {3110021013},
  pmid = {21905067}
}

@book{Chow2008,
  title = {Design and {{Analysis}} of {{Bioavailability}} and {{Bioequivalence Studies}}},
  author = {Chow, Shein-Chung and Liu, Jen-Pei},
  year = 2008,
  publisher = {Chapman \& Hall/CRC Press},
  address = {Boca Raton},
  urldate = {2021-04-19}
}

@article{Dawid1982,
  title = {The Well-Calibrated Bayesian},
  author = {Dawid, A. P.},
  year = 1982,
  journal = {Journal of the American Statistical Association},
  volume = {77},
  number = {379},
  pages = {605--610},
  issn = {1537274X},
  doi = {10.1080/01621459.1982.10477856},
  urldate = {2024-12-04}
}

@article{Dickey1970a,
  title = {The {{Weighted Likelihood Ratio}}, {{Sharp Hypotheses}} about {{Chances}}, the {{Order}} of a {{Markov Chain}}},
  author = {Dickey, James M. and Lientz, B. P.},
  year = 1970,
  journal = {Annals of Mathematical Statistics},
  volume = {41},
  number = {1},
  pages = {214--226},
  publisher = {Institute of Mathematical Statistics},
  issn = {0003-4851},
  doi = {10.1214/AOMS/1177697203},
  urldate = {2022-02-24}
}

@misc{europeanmedicinesagencyICHE20Adaptive2025,
  title = {{{ICH E20}} Adaptive Designs for Clinical Trials - {{Scientific}} Guideline \textbar{} {{European Medicines Agency}} ({{EMA}})},
  author = {{European Medicines Agency}},
  year = 2025,
  month = jun,
  publisher = {European Medicines Agency},
  urldate = {2025-12-23},
  langid = {english}
}

@article{Fayers2005,
  title = {Monitoring: {{Bayesian Data Monitoring}} in {{Clinical Trials}}},
  author = {Fayers, Peter M. and Ashby, Deborah and Parmar, Mahesh K.B.},
  year = 2005,
  month = aug,
  journal = {Tutorials in Biostatistics, Statistical Methods in Clinical Studies},
  volume = {1},
  pages = {335--352},
  publisher = {wiley},
  doi = {10.1002/0470023678.CH3B},
  urldate = {2022-02-02},
  isbn = {9780470023679}
}

@misc{FDA_ComplexInnovativeDesignsDecember2020,
  title = {Interacting with the {{FDA}} on {{Complex Innovative Trial Designs}} for {{Drugs}} and {{Biological Products}} - {{Guidance}} for {{Industry}}},
  author = {{U.S. Department of Health and Human Services} and {Food and Drug Administration} and {Center for Biologics Evaluation and Research} and {Center for Drug Evaluation and Research}},
  year = 2020,
  month = dec
}

@techreport{FDA_UseOfBayesianMethodologyJanuary2026,
  title = {Use of {{Bayesian Methodology}} in {{Clinical Trials}} of {{Drug}} and {{Biological Products}} - {{Guidance}} for {{Industry}}},
  author = {{U.S. Department of Health and Human Services Food and Drug Administration, Center for Drug Evaluation and Research (CDER), Center for Biologics Evaluation and Research (CBER)}},
  year = 2026,
  month = jan
}

@article{Ferguson2021,
  title = {Bayesian Interpretation of p Values in Clinical Trials},
  author = {Ferguson, John},
  year = 2021,
  month = sep,
  journal = {BMJ Evidence-Based Medicine},
  volume = {0},
  pages = {bmjebm-2020-111603},
  publisher = {Royal Society of Medicine},
  issn = {2515-446X},
  doi = {10.1136/BMJEBM-2020-111603},
  urldate = {2022-02-02},
  pmid = {34556541}
}

@article{Ferreira2021,
  title = {Bayesian Predictive Probabilities: A Good Way to Monitor Clinical Trials},
  author = {Ferreira, David and Ludes, Pierre Olivier and Diemunsch, Pierre and Noll, Eric and Torp, Klaus D. and Meyer, Nicolas},
  year = 2021,
  month = feb,
  journal = {British Journal of Anaesthesia},
  volume = {126},
  number = {2},
  pages = {550--555},
  publisher = {Elsevier},
  issn = {0007-0912},
  doi = {10.1016/J.BJA.2020.08.062},
  urldate = {2022-02-01},
  pmid = {33129491}
}

@article{gaoBayesianSequentialDecisionmaking2025,
  title = {Bayesian Sequential Decision-Making for Rare Disease Clinical Trials},
  author = {Gao, Yuan and Bai, Jianling and Chen, Feng},
  year = 2025,
  month = sep,
  journal = {Technology and Health Care},
  volume = {33},
  number = {5},
  pages = {2350--2370},
  publisher = {SAGE Publications},
  issn = {0928-7329},
  doi = {10.1177/09287329251344056},
  urldate = {2026-04-09},
  langid = {english}
}

@book{Good1983a,
  title = {Good {{Thinking}}: {{The Foundations}} of {{Probability}} and {{Its Applications}}},
  author = {Good, I.J.},
  year = 1983,
  publisher = {Minneapolis University Press},
  address = {Minneapolis}
}

@book{Grieve2022,
  title = {Hybrid Frequentist/{{Bayesian}} Power and {{Bayesian}} Power in Planning and Clinical Trials},
  author = {Grieve, Andrew P.},
  year = 2022,
  publisher = {Chapman \& Hall, CRC Press},
  address = {Boca Raton, FL},
  urldate = {2023-02-21},
  isbn = {978-1-032-11129-2}
}

@article{grieveIdleThoughtsWellcalibrated2016,
  title = {Idle Thoughts of a 'well-Calibrated' {{Bayesian}} in Clinical Drug Development},
  author = {Grieve, Andrew P.},
  year = 2016,
  month = mar,
  journal = {Pharmaceutical statistics},
  volume = {15},
  number = {2},
  pages = {96--108},
  publisher = {Pharm Stat},
  issn = {1539-1612},
  doi = {10.1002/PST.1736},
  pmid = {26799060}
}

@article{Gunel1974,
  title = {Bayes Factors for Independence in Contingency Tables},
  author = {Gunel, E. and Dickey, J.},
  year = 1974,
  journal = {Biometrika},
  volume = {61},
  number = {3},
  pages = {545--557},
  issn = {00063444},
  doi = {10.2307/2334738},
  urldate = {2022-02-24}
}

@article{hagarDesignBayesianClinical2026,
  title = {Design of {{Bayesian Clinical Trials With Clustered Data}}},
  author = {Hagar, Luke and Golchi, Shirin},
  year = 2026,
  month = mar,
  journal = {Stat Med},
  volume = {45},
  number = {6-7},
  pages = {e70488},
  issn = {0277-6715},
  doi = {10.1002/sim.70488},
  urldate = {2026-04-09},
  pmcid = {PMC12982163},
  pmid = {41820232}
}

@article{HeathEtAl2020,
  title = {Determining a {{Bayesian}} Predictive Power Stopping Rule for Futility in a Non-Inferiority Trial with Binary Outcomes},
  author = {Heath, Anna and Offringa, Martin and Pechlivanoglou, Petros and Rios, Juan David and Klassen, Terry P. and Poonai, Naveen and Pullenayegum, Eleanor},
  year = 2020,
  month = apr,
  journal = {Contemp Clin Trials Commun},
  volume = {18},
  pages = {100561},
  issn = {2451-8654},
  doi = {10.1016/j.conctc.2020.100561},
  urldate = {2026-03-27},
  pmcid = {PMC7153169},
  pmid = {32300671}
}

@article{hoekstraBayesianReanalysisNull2018,
  title = {Bayesian Reanalysis of Null Results Reported in Medicine: {{Strong}} yet Variable Evidence for the Absence of Treatment Effects},
  shorttitle = {Bayesian Reanalysis of Null Results Reported in Medicine},
  author = {Hoekstra, Rink and Monden, Rei and {van Ravenzwaaij}, Don and Wagenmakers, Eric-Jan},
  year = 2018,
  month = apr,
  journal = {PLoS One},
  volume = {13},
  number = {4},
  pages = {e0195474},
  issn = {1932-6203},
  doi = {10.1371/journal.pone.0195474},
  urldate = {2026-03-27},
  pmcid = {PMC5919013},
  pmid = {29694370}
}

@article{ionanBayesianMethodsHuman2023,
  title = {Bayesian {{Methods}} in {{Human Drug}} and {{Biological Products Development}} in {{CDER}} and {{CBER}}},
  author = {Ionan, Alexei C. and Clark, Jennifer and Travis, James and Amatya, Anup and Scott, John and Smith, James P. and Chattopadhyay, Somesh and Salerno, Mary Jo and Rothmann, Mark},
  year = 2023,
  journal = {Ther Innov Regul Sci},
  volume = {57},
  number = {3},
  pages = {436--444},
  issn = {2168-4790},
  doi = {10.1007/s43441-022-00483-0},
  urldate = {2025-12-12},
  pmcid = {PMC9718464},
  pmid = {36459346}
}

@article{Jamil2017,
  title = {Default ``{{Gunel}} and {{Dickey}}'' {{Bayes}} Factors for Contingency Tables},
  author = {Jamil, Tahira and Ly, Alexander and Morey, Richard D. and Love, Jonathon and Marsman, Maarten and Wagenmakers, Eric Jan},
  year = 2017,
  month = apr,
  journal = {Behavior Research Methods},
  volume = {49},
  number = {2},
  pages = {638--652},
  publisher = {Springer New York LLC},
  issn = {15543528},
  doi = {10.3758/S13428-016-0739-8/FIGURES/5},
  urldate = {2022-02-24},
  pmid = {27325166}
}

@book{Jeffreys1939,
  title = {Theory of {{Probability}}},
  author = {Jeffreys, Harold},
  year = 1939,
  publisher = {The Clarendon Press},
  address = {Oxford}
}

@article{jiangComparingBayesianEarly2020,
  title = {Comparing {{Bayesian}} Early Stopping Boundaries for Phase {{II}} Clinical Trials},
  author = {Jiang, Liyun and Yan, Fangrong and Thall, Peter F. and Huang, Xuelin},
  year = 2020,
  month = nov,
  journal = {Pharm Stat},
  volume = {19},
  number = {6},
  pages = {928--939},
  issn = {1539-1604},
  doi = {10.1002/pst.2046},
  urldate = {2026-03-27},
  pmcid = {PMC9149588},
  pmid = {32720462}
}

@article{KassRaftery1995,
  title = {Bayes Factors},
  author = {Kass, Robert E and Raftery, Adrian E},
  year = 1995,
  journal = {Journal of the American Statistical Association},
  volume = {90},
  number = {430},
  pages = {773--795}
}

@article{Kelter2020,
  title = {Bayesian Survival Analysis in {{STAN}} for Improved Measuring of Uncertainty in Parameter Estimates},
  author = {Kelter, Riko},
  year = 2020,
  journal = {Measurement: Interdisciplinary Research and Perspectives},
  volume = {18},
  number = {2},
  pages = {101--119},
  doi = {10.1080/15366367.2019.1689761}
}

@article{Kelter2020BayesianPosteriorIndices,
  title = {Analysis of {{Bayesian}} Posterior Significance and Effect Size Indices for the Two-Sample t-Test to Support Reproducible Medical Research},
  author = {Kelter, Riko},
  year = 2020,
  journal = {BMC Medical Research Methodology},
  volume = {20},
  number = {88},
  doi = {10.1186/s12874-020-00968-2}
}

@article{Kelter2021BMCHodgesLehmann,
  title = {Bayesian {{Hodges-Lehmann}} Tests for Statistical Equivalence in the Two-Sample Setting: {{Power}} Analysis, Type {{I}} Error Rates and Equivalence Boundary Selection in Biomedical Research},
  author = {Kelter, Riko},
  year = 2021,
  journal = {BMC Medical Research Methodology},
  volume = {21},
  number = {171},
  publisher = {BioMed Central},
  issn = {1471-2288},
  doi = {10.1186/s12874-021-01341-7}
}

@article{Kelter2022EvidenceValue,
  title = {The {{Evidence Interval}} and the {{Bayesian Evidence Value}} - {{On}} a Unified Theory for {{Bayesian}} Hypothesis Testing and Interval Estimation},
  author = {Kelter, R.},
  year = 2022,
  journal = {British Journal of Mathematical and Statistical Psychology},
  volume = {75},
  number = {3},
  pages = {550--592},
  doi = {10.1111/bmsp.12267}
}

@article{Kelter2023,
  title = {The {{Bayesian}} Simulation Study ({{BASIS}}) Framework for Simulation Studies in Statistical and Methodological Research},
  author = {Kelter, Riko},
  year = 2023,
  journal = {Biometrical Journal},
  pages = {2200095},
  publisher = {John Wiley \& Sons, Ltd},
  issn = {1521-4036},
  doi = {10.1002/BIMJ.202200095},
  urldate = {2023-01-26}
}

@article{Kelter2025,
  title = {The {{Calibrated Bayesian Hypothesis Test}} for {{Directional Hypotheses}} of the {{Odds Ratio}} in 2x2 {{Contingency Tables}}},
  author = {Kelter, Riko},
  year = 2025,
  journal = {Stat Biosci},
  volume = {17},
  number = {2},
  pages = {410--441},
  issn = {1867-1772},
  doi = {10.1007/s12561-024-09425-w},
  urldate = {2025-11-04},
  langid = {english}
}

@article{kelterBayesianGroupSequentialPredictive2024,
  title = {The {{Bayesian Group-Sequential Predictive Evidence Value Design}} for {{Phase II Clinical Trials}} with {{Binary Endpoints}}},
  author = {Kelter, Riko and Schnurr, Alexander},
  year = 2024,
  journal = {Statistics in Biosciences},
  number = {(online first)},
  pages = {1--37},
  publisher = {Springer},
  issn = {18671772},
  doi = {10.1007/s12561-024-09430-z}
}

@article{KelterPawel2025,
  title = {Bayesian {{Power}} and {{Sample Size Calculations}} for {{Bayes Factors}} in the {{Binomial Setting}}},
  author = {Kelter, Riko and Pawel, Samuel},
  year = 2025,
  journal = {arXiv preprint},
  eprint = {2502.02914},
  urldate = {2025-05-27},
  archiveprefix = {arXiv}
}

@misc{KelterPawelTwoStage2025,
  title = {The {{Bayesian}} Optimal Two-Stage Design for Clinical Phase {{II}} Trials Based on {{Bayes}} Factors},
  author = {Kelter, Riko and Pawel, Samuel},
  year = 2025,
  number = {arXiv:2511.23144},
  eprint = {2511.23144},
  primaryclass = {stat},
  publisher = {arXiv},
  doi = {10.48550/arXiv.2511.23144},
  urldate = {2025-12-31},
  archiveprefix = {arXiv}
}

@misc{kelterTwoArmTwoStage2026,
  title = {Power and {{Sample Size Calculations}} for {{Bayes Factors}} in Two-Arm Clinical {{Phase II Trials}} with Binary {{Endpoints}}},
  author = {Kelter, Riko},
  year = 2026,
  number = {arXiv:2603.01715},
  eprint = {2603.01715},
  primaryclass = {stat},
  publisher = {arXiv},
  doi = {10.48550/arXiv.2603.01715},
  urldate = {2026-03-24},
  archiveprefix = {arXiv}
}

@article{khannaRiociguatPatientsEarly2020,
  title = {Riociguat in Patients with Early Diffuse Cutaneous Systemic Sclerosis ({{RISE-SSc}}): Randomised, Double-Blind, Placebo-Controlled Multicentre Trial},
  shorttitle = {Riociguat in Patients with Early Diffuse Cutaneous Systemic Sclerosis ({{RISE-SSc}})},
  author = {Khanna, Dinesh and Allanore, Yannick and Denton, Christopher P. and Kuwana, Masataka and {Matucci-Cerinic}, Marco and Pope, Janet E. and Atsumi, Tatsuya and Be{\v c}v{\'a}{\v r}, Radim and Czirj{\'a}k, L{\'a}szl{\'o} and Hachulla, Eric and Ishii, Tomonori and Ishikawa, Osamu and Johnson, Sindhu R. and Langhe, Ellen De and Stagnaro, Chiara and Riccieri, Valeria and Schiopu, Elena and Silver, Richard M. and Smith, Vanessa and Steen, Virginia and Stevens, Wendy and Sz{\"u}cs, Gabriella and Truchetet, Marie-Elise and Wosnitza, Melanie and Laapas, Kaisa and Pena, Janethe de Oliveira and Yao, Zhen and Kramer, Frank and Distler, Oliver},
  year = 2020,
  month = may,
  journal = {Annals of the Rheumatic Diseases},
  volume = {79},
  number = {5},
  pages = {618--625},
  publisher = {BMJ Publishing Group Ltd},
  issn = {0003-4967, 1468-2060},
  doi = {10.1136/annrheumdis-2019-216823},
  urldate = {2025-12-11},
  chapter = {Systemic sclerosis},
  copyright = {\copyright{} Author(s) (or their employer(s)) 2020. Re-use permitted under CC BY. Published by BMJ.. https://creativecommons.org/licenses/by/4.0/This is an open access article distributed in accordance with the Creative Commons Attribution 4.0 Unported (CC BY 4.0) license, which permits others to copy, redistribute, remix, transform and build upon this work for any purpose, provided the original work is properly cited, a link to the licence is given, and indication of whether changes were made. See:~https://creativecommons.org/licenses/by/4.0/.},
  langid = {english},
  pmid = {32299845}
}

@article{Linde2020,
  title = {Baymedr: {{An R Package}} for the {{Calculation}} of {{Bayes Factors}} for {{Equivalence}}, {{Non-Inferiority}}, and {{Superiority Designs}}},
  author = {Linde, Maximilian and {van Ravenzwaaij}, Don},
  year = 2020,
  journal = {arXiv preprint: arXiv:1910.11616v1},
  eprint = {1910.11616v1},
  archiveprefix = {arXiv}
}

@article{Little2006,
  title = {Calibrated {{Bayes}}},
  author = {Little, Roderick J.},
  year = 2006,
  month = aug,
  journal = {The American Statistician},
  volume = {60},
  number = {3},
  pages = {213--223},
  publisher = {Taylor \& Francis},
  issn = {00031305},
  doi = {10.1198/000313006X117837},
  urldate = {2024-12-04}
}

@article{Makowski2019,
  title = {Indices of {{Effect Existence}} and {{Significance}} in the {{Bayesian Framework}}},
  author = {Makowski, Dominique and {Ben-Shachar}, Mattan S. and Chen, S. H. Annabel and L{\"u}decke, Daniel},
  year = 2019,
  journal = {Frontiers in Psychology},
  volume = {10},
  pages = {2767},
  publisher = {Frontiers},
  issn = {1664-1078},
  doi = {10.3389/fpsyg.2019.02767},
  urldate = {2020-01-28}
}

@article{Morris2019,
  title = {Using Simulation Studies to Evaluate Statistical Methods},
  author = {Morris, Tim P. and White, Ian R. and Crowther, Michael J.},
  year = 2019,
  journal = {Statistics in Medicine},
  volume = {38},
  number = {11},
  eprint = {1712.03198},
  pages = {2074--2102},
  publisher = {John Wiley \& Sons, Ltd},
  issn = {1097-0258},
  doi = {10.1002/SIM.8086},
  urldate = {2022-03-07},
  archiveprefix = {arXiv},
  pmid = {30652356}
}

@article{MuehlemannEtAl2023,
  title = {A {{Tutorial}} on {{Modern Bayesian Methods}} in {{Clinical Trials}}},
  author = {Muehlemann, Natalia and Zhou, Tianjian and Mukherjee, Rajat and Hossain, Munshi Imran and Roychoudhury, Satrajit and {Russek-Cohen}, Estelle},
  year = 2023,
  journal = {Ther Innov Regul Sci},
  volume = {57},
  number = {3},
  pages = {402--416},
  issn = {2168-4790},
  doi = {10.1007/s43441-023-00515-3},
  urldate = {2026-03-27},
  pmcid = {PMC10117244},
  pmid = {37081374}
}

@article{Neuenschwander2009,
  title = {A Note on the Power Prior},
  author = {Neuenschwander, Beat and Branson, Michael and Spiegelhalter, David J.},
  year = 2009,
  month = dec,
  journal = {Statistics in medicine},
  volume = {28},
  number = {28},
  pages = {3562--3566},
  publisher = {Stat Med},
  issn = {1097-0258},
  doi = {10.1002/SIM.3722},
  urldate = {2022-08-03},
  pmid = {19735071}
}

@misc{pawelBayesFactorGroup2026,
  title = {Bayes {{Factor Group Sequential Designs}}},
  author = {Pawel, Samuel and Held, Leonhard},
  year = 2026,
  month = jan,
  publisher = {Zenodo},
  doi = {10.5281/ZENODO.18160652},
  urldate = {2026-04-09},
  archiveprefix = {Zenodo},
  copyright = {Creative Commons Attribution 4.0 International}
}

@article{PawelHeld2025,
  title = {Closed-{{Form Power}} and {{Sample Size Calculations}} for {{Bayes Factors}}},
  author = {Pawel, Samuel and Held, Leonhard},
  year = 2025,
  month = apr,
  journal = {The American Statistician},
  pages = {1--15},
  publisher = {Taylor \& Francis},
  issn = {0003-1305},
  doi = {10.1080/00031305.2025.2467919},
  urldate = {2025-05-27}
}

@article{pourmohamadSequentialBayesFactors2023,
  title = {Sequential {{Bayes Factors}} for {{Sample Size Reduction}} in {{Preclinical Experiments}} with {{Binary Outcomes}}},
  author = {Pourmohamad, Tony and Wang, Chenguang},
  year = 2023,
  journal = {Statistics in Biopharmaceutical Research},
  volume = {15},
  number = {4},
  pages = {706--715},
  publisher = {Taylor \& Francis},
  issn = {19466315},
  doi = {10.1080/19466315.2022.2123386}
}

@article{Rouder2009,
  title = {Bayesian t Tests for Accepting and Rejecting the Null Hypothesis},
  author = {Rouder, Jeffrey N. and Speckman, Paul L. and Sun, Dongchu and Morey, Richard D. and Iverson, Geoffrey},
  year = 2009,
  journal = {Psychonomic Bulletin and Review},
  volume = {16},
  number = {2},
  pages = {225--237},
  issn = {10699384},
  doi = {10.3758/PBR.16.2.225},
  isbn = {1069-9384(Print)},
  pmid = {19293088}
}

@article{Schonbrodt2017,
  title = {Sequential Hypothesis Testing with {{Bayes}} Factors: {{Efficiently}} Testing Mean Differences},
  author = {Sch{\"o}nbrodt, Felix D. and Wagenmakers, Eric Jan and Zehetleitner, Michael and Perugini, Marco},
  year = 2017,
  month = jun,
  journal = {Psychological methods},
  volume = {22},
  number = {2},
  pages = {322--339},
  publisher = {Psychol Methods},
  issn = {1939-1463},
  doi = {10.1037/MET0000061},
  urldate = {2022-12-01},
  pmid = {26651986}
}

@article{sekulovskiGoodCheckBayes2024,
  title = {A {{Good}} Check on the {{Bayes}} Factor},
  author = {Sekulovski, Nikola and Marsman, Maarten and Wagenmakers, Eric-Jan},
  year = 2024,
  journal = {Behav Res Methods},
  volume = {56},
  number = {8},
  pages = {8552--8566},
  issn = {1554-351X},
  doi = {10.3758/s13428-024-02491-4},
  urldate = {2026-03-24},
  pmcid = {PMC11525426},
  pmid = {39231912}
}

@article{shenBayesianGroupSequential2022,
  title = {Bayesian Group Sequential Designs for Cluster-Randomized Trials},
  author = {Shen, Junwei and Golchi, Shirin and Moodie, Erica EM and Benrimoh, David},
  year = 2022,
  journal = {Stat},
  volume = {11},
  number = {1},
  pages = {e487},
  issn = {2049-1573},
  doi = {10.1002/sta4.487},
  urldate = {2026-04-09},
  copyright = {\copyright{} 2022 John Wiley \& Sons Ltd.},
  langid = {english}
}

@article{siepeSimulationStudiesMethodological2024,
  title = {Simulation Studies for Methodological Research in Psychology: {{A}} Standardized Template for Planning, Preregistration, and Reporting.},
  author = {Siepe, Bj{\"o}rn S. and Barto{\v s}, Franti{\v s}ek and Morris, Tim P. and Boulesteix, Anne-Laure and Heck, Daniel W. and Pawel, Samuel},
  year = 2024,
  month = nov,
  journal = {Psychological Methods},
  issn = {1939-1463},
  doi = {10.1037/MET0000695}
}

@article{Simon1989,
  title = {Optimal Two-Stage Designs for Phase {{II}} Clinical Trials},
  author = {Simon, Richard},
  year = 1989,
  journal = {Controlled clinical trials},
  volume = {10},
  number = {1},
  pages = {1--10},
  publisher = {Control Clin Trials},
  issn = {0197-2456},
  doi = {10.1016/0197-2456(89)90015-9},
  urldate = {2022-02-01},
  pmid = {2702835}
}

@book{Spiegelhalter2004,
  title = {Bayesian Approaches to Clinical Trials and Health-Care Evaluation},
  author = {Spiegelhalter, D. J. and Abrams, K. R. (Keith R.) and Myles, Jonathan P.},
  year = 2004,
  publisher = {Wiley},
  address = {New York},
  urldate = {2019-12-11},
  isbn = {978-0-470-09260-6}
}

@book{Sprenger2019,
  title = {Bayesian {{Philosophy}} of {{Science}}},
  author = {Sprenger, Jan and Hartmann, Stephan},
  year = 2019,
  month = aug,
  journal = {Bayesian Philosophy of Science},
  publisher = {Oxford University Press},
  doi = {10.1093/oso/9780199672110.001.0001}
}

@article{Stallard2020,
  title = {Comparison of {{Bayesian}} and Frequentist Group-Sequential Clinical Trial Designs},
  author = {Stallard, Nigel and Todd, Susan and Ryan, Elizabeth G. and Gates, Simon},
  year = 2020,
  journal = {BMC Medical Research Methodology},
  volume = {20},
  number = {1},
  pages = {1--14},
  publisher = {BioMed Central Ltd.},
  issn = {14712288},
  doi = {10.1186/S12874-019-0892-8/FIGURES/4},
  urldate = {2023-05-17},
  pmid = {31910813}
}

@article{Stefan2022,
  title = {A {{Two-Stage Bayesian Sequential Assessment}} of {{Exploratory Hypotheses}}},
  author = {Stefan, Angelika M. and Lengersdorff, Lukas L. and Wagenmakers, Eric-Jan},
  year = 2022,
  journal = {Collabra: Psychology},
  volume = {8},
  number = {1},
  publisher = {University of California Press},
  issn = {2474-7394},
  doi = {10.1525/COLLABRA.40350},
  urldate = {2022-12-01}
}

@article{Thall1994,
  title = {Practical {{Bayesian Guidelines}} for {{Phase IIB Clinical Trials}}},
  author = {Thall, Peter F. and Simon, Richard},
  year = 1994,
  month = jun,
  journal = {Biometrics},
  volume = {50},
  number = {2},
  pages = {337},
  publisher = {JSTOR},
  issn = {0006341X},
  doi = {10.2307/2533377},
  urldate = {2022-11-22},
  pmid = {7980801}
}

@article{VandeSchoot2021,
  title = {Bayesian Statistics and Modelling},
  author = {{Van de Schoot}, Rens and Depaoli, Sarah and King, Ruth and Kramer, Bianca and M{\"a}rtens, Kaspar and Tadesse, Mahlet G. and Vannucci, Marina and Gelman, Andrew and Veen, Duco and Willemsen, Joukje and Yau, Christopher},
  year = 2021,
  month = jan,
  journal = {Nature Reviews Methods Primers 2021 1:1},
  volume = {1},
  number = {1},
  pages = {1--26},
  publisher = {Nature Publishing Group},
  issn = {2662-8449},
  doi = {10.1038/s43586-020-00001-2},
  urldate = {2022-03-08},
  isbn = {0123456789}
}

@book{WassmerBrannath2016,
  title = {Group {{Sequential}} and {{Confirmatory Adaptive Designs}} in {{Clinical Trials}}},
  author = {Wassmer, Gernot and Brannath, Werner},
  year = 2016,
  publisher = {Springer International Publishing Switzerland},
  address = {Cham},
  doi = {10.1007/978-3-319-32562-0},
  urldate = {2023-05-17},
  isbn = {978-3-319-32560-6}
}

@article{Zhou2023,
  title = {On {{Bayesian Sequential Clinical Trial Designs}}},
  author = {Zhou, Tianjian and Ji, Yuan},
  year = 2023,
  journal = {The New England Journal of Statistics in Data Science},
  volume = {0},
  eprint = {2112.09644},
  pages = {1--16},
  publisher = {New England Statistical Society},
  issn = {2693-7166},
  doi = {10.51387/23-NEJSDS24},
  urldate = {2023-05-17},
  archiveprefix = {arXiv}
}

@article{zhuBayesianSequentialDesign2019,
  title = {A {{Bayesian Sequential Design}} for {{Clinical Trials}} with {{Time-to-Event Outcomes}}},
  author = {Zhu, Lin and Yu, Qingzhao and Mercante, Donald E.},
  year = 2019,
  journal = {Stat Biopharm Res},
  volume = {11},
  number = {4},
  pages = {387--397},
  issn = {1946-6315},
  doi = {10.1080/19466315.2019.1629996},
  urldate = {2026-04-09},
  pmcid = {PMC7100880},
  pmid = {32226580}
}

\end{document}